\documentclass[11pt,oneside,reqno]{amsart}
\usepackage{graphicx}
\usepackage{nicefrac}
\usepackage{amsmath,bm}
\usepackage{mathtools}
\usepackage{geometry}
\usepackage[nodisplayskipstretch]{setspace}
\usepackage{subcaption}
\usepackage{hyperref}
\usepackage{float}
\usepackage{epstopdf}
\usepackage{interval}
\usepackage{natbib}
\usepackage{afterpage}
\usepackage{algorithm}
\usepackage{amsaddr}
\newtheorem{thm}{Theorem}[section]

\theoremstyle{definition}

\theoremstyle{remark}

\numberwithin{equation}{section}

\newenvironment{nouppercase}{%
  \let\uppercase\relax%
  \renewcommand{\uppercasenonmath}[1]{}}{}
\begin{document}

\title[]{The Extended Power Distribution: A new distribution on (0, 1)}%
\author{C. ~E. Ogbonnaya, S. ~P. Preston, A. ~T. ~A. Wood}%
\address{School of Mathematical Sciences, University of Nottingham}%
\thanks{pmxceog@nottingham.ac.uk}%


\begin{abstract}
We propose a two-parameter bounded probability distribution called the extended power distribution. This distribution on $(0, 1)$ is similar to the beta  distribution, however there are some advantages which we explore. We define the moments and quantiles of this distribution and show that it is possible to give an $r$-parameter extension of this distribution ($r>2$). We also consider its complementary distribution and show that it has some flexibility advantages over the Kumaraswamy and beta distributions. This distribution can be used as an alternative to the Kumaraswamy distribution since it has a closed form for its cumulative function. However, it can be fitted to data where there are some samples that are exactly equal to 1, unlike the Kumaraswamy and beta distributions which cannot be fitted to such data or may require some censoring. Applications considered show the extended power distribution performs favourably against the Kumaraswamy distribution in most cases.
\end{abstract}
\begin{nouppercase}
\maketitle
\end{nouppercase}
\section{Introduction}
In this work, we propose an interesting two-parameter probability distribution with bounded support and propose this as an alternative to the Kumaraswamy and beta distributions. The proposed distribution is bounded on $(0, 1)$, just like the beta and Kumaraswamy distributions. However, the extended power distribution has some advantage over the beta distribution, since its cumulative distribution can be obtained in closed form. We will explore properties of this distribution such as moments, quantiles and cumulative distribution. We even go further to give a closed form for its complementary distribution as discussed in \cite{jones2002complementary} and \cite{jones2009kumaraswamy} for the beta and Kumaraswamy distributions respectively.

The extended power distribution is an extension of the power function distribution (which is a special case of the beta distribution). However, it has the advantage of being easily extendable to a multi-parameter case, which add extra flexibility when fitting to observed samples. We can also easily obtain the cumulative distribution function of the generalised extended power distribution in closed form unlike the generalised beta distribution.

Other distributions with bounded support have been investigated in statistical literature, such as the beta distribution, truncated normal distribution, log-Lindley distribution (\cite{gomez2014log}) and the Kumaraswamy distribution (\cite{kumaraswamy1980generalized}). The beta distribution has been used to model data arising from distribution of proportions and is widely used in bayesian analysis as a conjugate prior for sampling proportions from the binomial distribution. A beta regression model has been proposed by \cite{ferrari2004beta} for modelling responses that are proportions. Using the idea of the uniform distribution, \cite{jones2004families} and \cite{eugene2002beta} have proposed generating a new class of distributions from the beta distribution with the shape parameters controlling asymmetry. In \cite{eugene2002beta}, a new distribution called the beta-normal distribution was proposed and other properties such as moments were explored. Using the idea of \cite{eugene2002beta}, other distributions arising from the beta distribution have been proposed such as the beta-exponential distribution (\cite{nadarajah2006beta}), beta-Gumbel distribution (\cite{nadarajah2004beta}), beta generalised exponential distribution (\cite{barreto2010beta}), beta-Pareto distribution (\cite{akinsete2008beta}), beta linear failure rate distribution (\cite{jafari2012beta}) among others. In \cite{jones2002complementary}, a new distribution arising from the quantile function of the beta distribution named the complementary beta distribution is proposed. However, a drawback of the beta distribution is the non-availability of its cumulative distribution in closed form. To deal with this \cite{kumaraswamy1980generalized} proposed a double bounded distribution (renamed Kumaraswamy distribution by \cite{jones2009kumaraswamy}). This distribution was originally proposed for modelling data in the field of hydrology but later suggested as an alternative to the beta distribution with a closed form for its cumulative distribution and a simple density function without any special functions. A new class of distributions arising from the Kumaraswamy distribution was proposed by \cite{cordeiro2011new}. These are  sometimes called the Kumaraswamy-G distribution. Some new distributions proposed include the Kumaraswamy Weibull distribution (\cite{cordeiro2010kumaraswamy}), Kumaraswamy Gumbel distribution (\cite{cordeiro2012kumaraswamy}) and the Kumaraswamy generalised gamma distribution (\cite{de2011kumaraswamy}) among others. However, the Kumaraswamy distribution (just like the beta distribution) is unable to fit data in which some sample points are exactly $1$. The extended power distribution has some interesting advantages as an alternative to the beta distribution with its simple qunatile function and interesting complementary distribution.

In this work, we propose a new bounded distribution with  a closed form cumulative distribution function and with additional flexibility through generalisation. We will calculate moments of this distribution for the two parameter case and give its quantile function to enable simulations. We also show that for the generalised case with multiple parameters, simulation simply involves finding the feasible solution to some polynomial equation. We use applications to show that the extended power distribution performs favourably against the Kumaraswamy distribution. In section 2, we explore the origin and basic properties of the extended power distribution as well as special cases of the distribution such as the linear failure rate distribution, exponential and Raleigh distribution. In section 3, we calculate the moments and quantile of this distribution and give a procedure for calculating the maximum likelihood estimates of the parameters. We also explore the distribution of order statistics for the minimum and maximum and give a closed form for the generalised extended power distribution as well as discuss its basic properties. In section 4, we propose the complementary extended power distribution using the quantile distribution of the extended power distribution. We also obtain the moments, quantiles and give special cases of the complementary distribution. Conclusion and further discussions are given in section 5.

\section{Basics and special cases}
The name "extended power distribution" is obtained from the fact that the cumulative distribution function of the extended power function is derived from an extension of the power function. This was motivated by extension of the single parameter power warping function to a warping function with $r$ parameters (the warping functions are used in functional data analysis for aligning curves). The power function is given by
\begin{equation}\label{eqn:eqn2}
  G(t)=t^{\alpha}=\exp(\alpha \log (t)).
\end{equation}
We extended equation \ref{eqn:eqn2} in powers of $\log(t)$ and the two parameter case is what we have as the cumulative distribution function of the extended power function which is given in equation \ref{eqn:eqn3}
\begin{equation}\label{eqn:eqn3}
  F(t)=\exp\{\alpha_{0} \log(t)-\alpha_{1}(\log (t))^{2}\}.
\end{equation}
The probability density function of the extended power distribution (EPD) with parameters $\alpha_{0}$ and $\alpha_{1}$ is given as follows
\begin{equation}\label{eqn:eqn1}
  f(t)=\left\{\frac{\alpha_{0}-2\alpha_{1}\log (t)}{t}\right\}\exp\left\{\alpha_{0}\log (t) -\alpha_{1}(\log (t))^{2}\right\}, \quad t\in (0, 1).
\end{equation}
The shape parameters for this distribution satisfy $\alpha_{0}>0$ and $\alpha_{1}\geq 0$.
An implication of the relationship between the power function and the extended power distribution is that for $\alpha_{1}=0$, the extended power function reduces to a special case of beta distribution.
\begin{figure}[ht]
      \centering
    \includegraphics[scale=0.8]{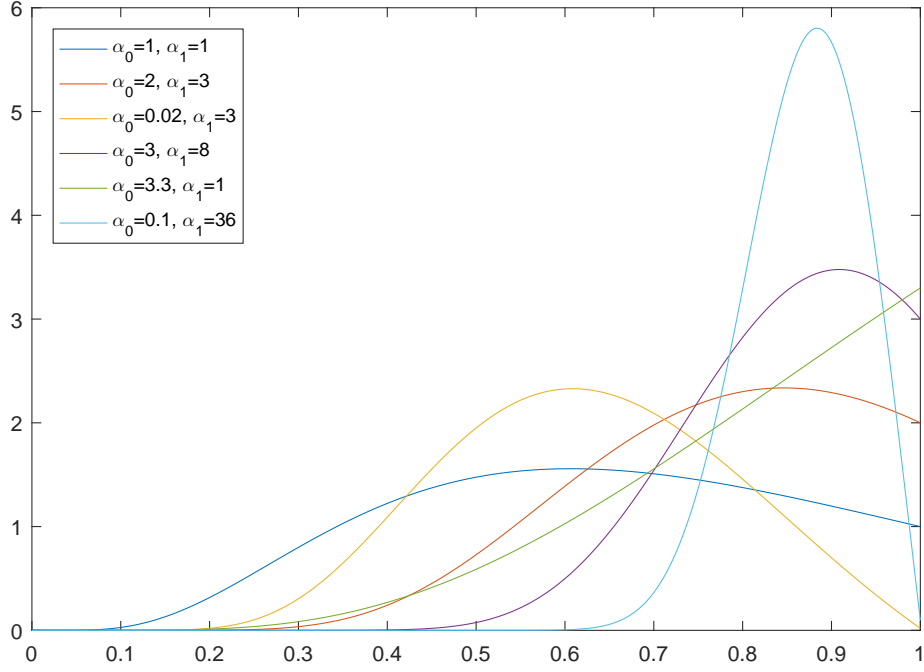}
    \caption{Plots of the extended power distribution for different values of $\alpha_{0}$ and $\alpha_{1}$}
    \label{fig:fig1}
  \end{figure}

If $\alpha_{1}=0$, then
 \begin{equation}\label{eqn:eqn4}
f(t)=\alpha_{0}t^{\alpha_{0}-1}
 \end{equation}
 which is a special case of the beta distribution with $\beta=1$. This special case is in fact the power function distribution, and is obtainable from the Kumaraswamy distribution by setting $\beta=1$.  Recall that the density function for the beta distribution is
 \begin{equation*}
   g(t)=\frac{t^{\alpha-1}(1-t)^{\beta-1}}{B(\alpha, \beta)}
 \end{equation*}
 where $B(\alpha, \beta)$ is beta function. The density function for the Kumaraswamy distribution is
\begin{equation*}
   g(t)=\alpha \beta t^{\alpha-1}(1-t^{\alpha})^{\beta-1}
 \end{equation*}
For $\alpha_{0}=1$ and $\alpha_{1}=0$, the extended power distribution reduces to the uniform distribution on $(0, 1)$. In a similar manner, the Beta(1, 1) and Kumaraswamy(1, 1) gives the uniform on $(0, 1)$ (see \cite{jones2009kumaraswamy}). If T follows the extended power distribution with parameters $\alpha_{0}$ and $\alpha_{1}$, then $V=-\log (T)$ is a random variable from the linear failure rate distribution (\cite{bain1974analysis}, \cite{sarhan2009generalized}) with density function
\begin{equation*}
  f(v)=(\alpha_{0}+2\alpha_{1}v)\exp\{-\alpha_{0}v-\alpha_{1}v^{2}\}, \quad 0<v<\infty.
\end{equation*}
Properties of this distribution including moments and quantiles have been studied by \cite{sen1995inference} and \cite{sen2006linear}.
For $\alpha_{1}=0$, we have
\begin{equation*}
  f(v)=\alpha_{0}\exp\{-\alpha_{0}v\}, \quad 0<v<\infty
\end{equation*}
which is an exponential distribution with parameter $\alpha_{0}$. If we allow $\alpha_{0}=0$ in the linear failure rate distribution, then $V$ reduces to a random variable from the Raleigh distribution with scale parameter $\sqrt{\frac{1}{2\alpha_{1}}}$ and density
\begin{equation*}
  f(v)=2\alpha_{1}v\exp\{-\alpha_{1}v^{2}\}, \quad 0<v<\infty.
\end{equation*}
As stated earlier, an advantage of the extended power distribution over the beta distribution is that we have its cumulative distribution function in closed form. With a closed form for the cumulative distribution and invertibility, it is possible to easily use the probability integral transform for simulation. If we define the $U\sim U(0, 1)$, then equating $F(T)=U$ from equation \ref{eqn:eqn3}, we have
\begin{equation*}
  \alpha_{1}(\log(T))^{2}-\alpha_{0}\log(T)+\log(U)=0
\end{equation*}
hence,
\begin{align}\label{eqn:eqn31}
  T=\begin{cases}
      \exp\left\{\frac{\alpha_{0}-\sqrt{\alpha_{0}^{2}-4\alpha_{1}\log(U)}}{2\alpha_{1}}\right\}, & \mbox{if } \alpha_{1}\neq 0 \\
      U^{\frac{1}{\alpha_{0}}}, & \mbox{otherwise}.
    \end{cases}
\end{align}
This random variable generator like that specified for the Kumaraswamy distribution (as mentioned in \cite{jones2009kumaraswamy}) is less complicated than those required to simulate from the beta distribution. To simulate a random variate T from the extended power distribution, we simply simulate U from $U(0, 1)$ and evaluate equation \ref{eqn:eqn31}. The limiting behaviour of the extended power distribution is as follows
\begin{eqnarray*}
  \lim_{t \to 0}\frac{f(t)}{t^{\alpha_{0}-1}} &=& 0 \\
  \lim_{t \to 1}f(t) &=&\alpha_{0}.
\end{eqnarray*}
\section{Moments, Quantiles and Estimators}
In this section, we will estimate some relevant quantities related to the extended power distribution. An interesting property of the extended power distribution is that we can estimate quantiles in nice closed form without any special functions as against the beta distribution whose median requires special functions. However, the moments of the extended power distribution are a bit more complicated than those of the beta and Kumaraswamy distributions. We need the complementary error integral functions (sometimes denoted by erfc(.)) to specify the moments of the extended power distribution.
\subsection{Moments}
We can derive a general formula for the kth moment of the extended power distribution.
\begin{thm}\label{cor1}
The kth moment of the extended power distribution is given as
\begin{equation}\label{eqn:eqn7}
E(T^{k})=1-\frac{k}{2}\sqrt{\frac{\pi}{\alpha_{1}}}\exp\left\{\frac{(\alpha_{0}+k)^{2}}{4\alpha_{1}}\right\}\text{erfc}\bigg(\frac{\alpha_{0}+k}{2\sqrt{\alpha_{1}}}\bigg).
\end{equation}
\end{thm}
\begin{proof}

\begin{equation*}
E(T^{k})=\int_{0}^{1}t^{k}f(t)dt=\int_{0}^{1}t^{k-1}(\alpha_{0}-2\alpha_{1}\log(t))\exp\left\{\alpha_{0}\log(t)-\alpha_{1}(\log(t))^{2}\right\}dt
\end{equation*}
Defining $u=\alpha_{0}\log(t)-\alpha_{1}(\log(t))^{2}$, we have
\begin{equation*}
E(T^{k})=\exp\left\{\frac{\alpha_{0}k}{2\alpha_{1}}\right\}\int_{-\infty}^{0}\exp\left\{\frac{k(\alpha_{0}^{2}-4\alpha_{1}u)^{1/2}}{-2\alpha_{1}}. \right\}\exp\{u\}du.
\end{equation*}
Let $v=k(\alpha_{0}^{2}-4\alpha_{1}u)^{1/2}$, this implies
\begin{equation*}
E(T^{k})=\frac{\exp\left\{\frac{(\alpha_{0}+k)^{2}}{4\alpha_{1}}\right\}}{2\alpha_{1}k^{2}}\int_{\alpha_{0}k}^{\infty}v\exp\left\{-\frac{(v+k^{2})^{2}}{4\alpha_{1}k^{2}}\right\}dv.
\end{equation*}
Therefore,
\begin{equation*}
E(T^{k})=1-\frac{k}{2}\sqrt{\frac{\pi}{\alpha_{1}}}\exp\left\{\frac{(\alpha_{0}+k)^{2}}{4\alpha_{1}}\right\}\text{erfc}\bigg(\frac{\alpha_{0}+k}{2\sqrt{\alpha_{1}}}\bigg).
\end{equation*}
\end{proof}
From theorem \ref{cor1}, we have the following results for $E(T)$ and $Var(T)$
\begin{align*}
E(T)=&1-\frac{1}{2}\sqrt{\frac{\pi}{\alpha_{1}}}\exp\left\{\frac{(\alpha_{0}+1)^{2}}{4\alpha_{1}}\right\}\text{erfc}\bigg(\frac{\alpha_{0}+1}{2\sqrt{\alpha_{1}}}\bigg)\\
Var(T)=&\bigg( \sqrt{\frac{\pi}{\alpha_{1}}}\exp\left\{\frac{(\alpha_{0}+1)^{2}}{4\alpha_{1}}\right\}\text{erfc}\bigg(\frac{\alpha_{0}+1}{2\sqrt{\alpha_{1}}}\bigg)\bigg)\bigg( 1-\frac{1}{4}\sqrt{\frac{\pi}{\alpha_{1}}}\exp\left\{\frac{(\alpha_{0}+1)^{2}}{4\alpha_{1}}\right\}\text{erfc}\bigg(\frac{\alpha_{0}+1}{2\sqrt{\alpha_{1}}}\bigg)\bigg)\\
&-\sqrt{\frac{\pi}{\alpha_{1}}}\exp\left\{\frac{(\alpha_{0}+2)^{2}}{4\alpha_{1}}\right\}\text{erfc}\bigg(\frac{\alpha_{0}+2}{2\sqrt{\alpha_{1}}}\bigg)
\end{align*}
where $\text{erfc}(x)=2(1-\Phi(x\sqrt{2}))$ and $\Phi(.)$ is the normal CDF.
\subsection{Quantiles and Mode}
The $p$th quantile function for the extended power distribution is easily obtainable from the cumulative distribution function and is given as
\begin{equation}\label{eqn:eqn8}
Q_{p}=\exp\left\{\frac{\alpha_{0}-\sqrt{\alpha_{0}^{2}-4\alpha_{1}\log(p)}}{2\alpha_{1}}\right\}.
\end{equation}
Estimating the median in closed form is straightforward from equation \ref{eqn:eqn8} and it can be written as shown in equation \ref{eqn:eqn9}.
\begin{equation}\label{eqn:eqn9}
Q_{0.5}=\exp\left\{\frac{\alpha_{0}-\sqrt{\alpha_{0}^{2}-4\alpha_{1}\log(0.5)}}{2\alpha_{1}}\right\}
\end{equation}
This expression for the median  is obviously an improvement on the beta distribution, where the median is expressed in terms of the incomplete beta function. Since the first derivative of the density ($f(t)$) is easily obtainable, we can calculate the mode of the extended power distribution in closed form.
The mode for the extended power distribution is given as
\begin{equation}\label{eqn:eqn10}
M=\exp\left\{ \frac{(2\alpha_{0}-1)-\sqrt{1+8\alpha_{1}}}{4\alpha_{1}}\right\}.
\end{equation}
\subsection{Maximum Likelihood Estimators}
Like in the beta and Kumaraswamy distributions, there is no simple form for the maximum likelihood estimators (MLEs) of $\alpha_{0}$ and $\alpha_{1}$. However, we can use a non-linear optimisation procedure to estimate these parameters numerically.
Given n random samples from the extended power distribution $t_{1}, t_{2}, \ldots, t_{n}$, the log-likelihood function is
\begin{equation*}
\ell(\alpha_{0}, \alpha_{1})=\sum_{i=1}^{n}\log (\alpha_{0}-2\alpha_{1}\log t_{i})-\sum_{i=1}^{n}\log t_{i}+\alpha_{0}\sum_{i=1}^{n}\log t_{i}-\alpha_{1}\sum_{i=1}^{n}(\log t_{i})^{2}.
\end{equation*}
Differentiating $\ell(\alpha_{0}, \alpha_{1})$ w.r.t $\alpha_{0}$ and $\alpha_{1}$, we obtain the system of equations
\begin{eqnarray}
  \frac{\partial \ell(\alpha_{0}, \alpha_{1}) }{\partial \alpha_{0}} &=& \sum_{i=1}^{n}\frac{1}{(\alpha_{0}-2\alpha_{1}\log t_{i})}+\sum_{i=1}^{n}\log t_{i}=0 \\ \label{eqn:eqn11}
\frac{\partial \ell(\alpha_{0}, \alpha_{1}) }{\partial \alpha_{1}} &=& 2\sum_{i=1}^{n}\frac{\log t_{i}}{(\alpha_{0}-2\alpha_{1}\log t_{i})}+\sum_{i=1}^{n}(\log t_{i})^{2}=0. \label{eqn:eqn12}
\end{eqnarray}
The observed Fisher's information matrix is
\begin{equation*}
  \bm{H}=\left[
     \begin{array}{cc}
       \sum_{i=1}^{n}\frac{1}{(\hat{\alpha}_{0}-2\hat{\alpha}_{1}\log t_{i})^{2}} &  -2\sum_{i=1}^{n}\frac{\log t_{i}}{(\hat{\alpha}_{0}-2\hat{\alpha}_{1}\log t_{i})^{2}} \\
        -2\sum_{i=1}^{n}\frac{\log t_{i}}{(\hat{\alpha}_{0}-2\hat{\alpha}_{1}\log t_{i})^{2}} & 4\sum_{i=1}^{n}\frac{(\log t_{i})^{2}}{(\hat{\alpha}_{0}-2\hat{\alpha}_{1}\log t_{i})^{2}} \\
     \end{array}
   \right].
\end{equation*}

Estimating the method of moments estimators will be more complicated, because the parameters of the distribution are contained in  special functions of the population moment.
In table \ref{tab:table1}, we simulate random samples using the probability integral transform and use the MLE method to estimate the parameters of these samples. In each case, 5000 random samples were generated with specified parameter values for $\alpha_{0}$ and $\alpha_{1}$ and estimated parameters are compared to the actual parameter values.
\begin{table}[h!]
\begin{tabular}{|p{4.8cm}||p{5.9cm}|}
 \hline
  Actual Parameters $(\alpha_{0}, \alpha_{1})$  &Maximum Likelihood estimates $(\hat{\alpha}_{0}, \hat{\alpha}_{1})$\\
 \hline
 (2, 1)& (2.0042, 1.0088)\\
 \hline
 (1, 1)&(1.0110, 1.0022)\\
 \hline
  (1.2, 3.3)&(1.2093, 3.3191)\\
\hline
 (0.02, 5)&(0.0174, 5.0162)\\
\hline
 (3, 8)&(3.0528, 8.0279)\\
\hline
 (0.8, 5)&(0.8301, 4.9695)\\
\hline
 (0.8, 25)&(0.8555, 25.5528)\\
\hline
 (1, 0.01)&(1.0047, 0.0063)\\
\hline
\end{tabular}
\caption{MLE for different simulated samples}
\label{tab:table1}
\end{table}
\subsection{Distribution of Order Statistics}
Let $T_{(1)}\leq T_{(2)} \leq T_{(3)} \leq \ldots \leq T_{(n)}$ be the order statistics of a random sample of size n from the extended power distribution with parameters $\alpha_{0}$ and $\alpha_{1}$ (EPD($\alpha_{0}$, $\alpha_{1}$)). Then the minimum, $T_{(1)}$ has density function
\begin{equation}\label{eqn:eqn28}
f(t_{(1)})=\left\{\frac{\alpha_{0}n-2\alpha_{1}n\log (t)}{t}\right\}\exp\left\{\alpha_{0}\log (t) -\alpha_{1}(\log (t))^{2}\right\}\bigg[1- \exp\left\{\alpha_{0}\log (t) -\alpha_{1}(\log (t))^{2}\right\}\bigg]^{n-1}.
\end{equation}
The minimum has the Kumaraswamy-G (Kw-G) distribution with parameters $a=1$ and $b=n$. The Kw-G distribution was introduced by \cite{cordeiro2011new} (motivated by \cite{jones2004families} work on distributions arising from beta distribution)and has density function
\begin{equation}\label{eqn:eqn29}
f(t)=abg(t)G(t)^{a-1}\bigg[1-G(t)^{a}\bigg]^{b-1}
\end{equation}
where $G(t)$ is the parent continuous cumulative distribution function.
Similarly, the maximum $T_{(n)}$ has density function
 \begin{equation}\label{eqn:eqn30}
 f(t_{(n)})=\left\{\frac{\alpha_{0}n-2\alpha_{1}n\log (t)}{t}\right\}\exp\left\{\alpha_{0}n\log (t) -\alpha_{1}n(\log (t))^{2}\right\}
 \end{equation}
 which is an extended power distribution with parameters $\alpha_{0}n$ and $\alpha_{1}n$ (EPD ($\alpha_{0}n$, $\alpha_{1}n$)).
\subsection{A Generalisation}
An important advantage of the extended power distribution is that it can be easily generalised  and will have a similar form to the two parameter case. Generalisations of the beta and Kumaraswamy distributions have been studied by \cite{gordy1998generalization}, \cite{nadarajah2003generalized}, \cite{mcdonald1995generalization} and these usually have complicated forms with cumulative distribution needing special functions (which makes simulations more complex and requires algorithms like rejection sampling). The generalised beta distribution has density function
\begin{equation}\label{eqn:eqn27}
h(t)=\frac{\mid a\mid t^{ap-1}(1-(1-c)(t/b)^{a})^{q-1}}{b^{ap}B(p, q)(1+c(t/b)^{a})^{p+q}}, \quad 0<t^{a}<b^{a}/(1-c).
\end{equation}
We can generalise the extended power distribution to have $r$ parameters and the density function is given as follows
\begin{equation}\label{eqn:eqn22}
f(t)=\frac{1}{t}\sum_{h=1}^{r}\alpha_{h-1}(\log t)^{h-1}h(-1)^{h-1}\exp\left\{\sum_{h=1}^{r}(-1)^{h-1}\alpha_{h-1}(\log t)^{h}\right\}, \quad t \in (0, 1).
\end{equation}
The $r$ parameters of this density function are $\alpha_{0}>0, \alpha_{1}\geq0, \alpha_{2}\geq0, \ldots, \alpha_{r-1}\geq0 $ are determine the shape of the distribution. Figure \ref{fig:fig6} gives plot of the density function of this distribution for the three-parameter and four-parameter cases.
\begin{figure}[ht]
    \begin{subfigure}[b]{0.5\linewidth}
 \centering
    \includegraphics[scale=0.42]{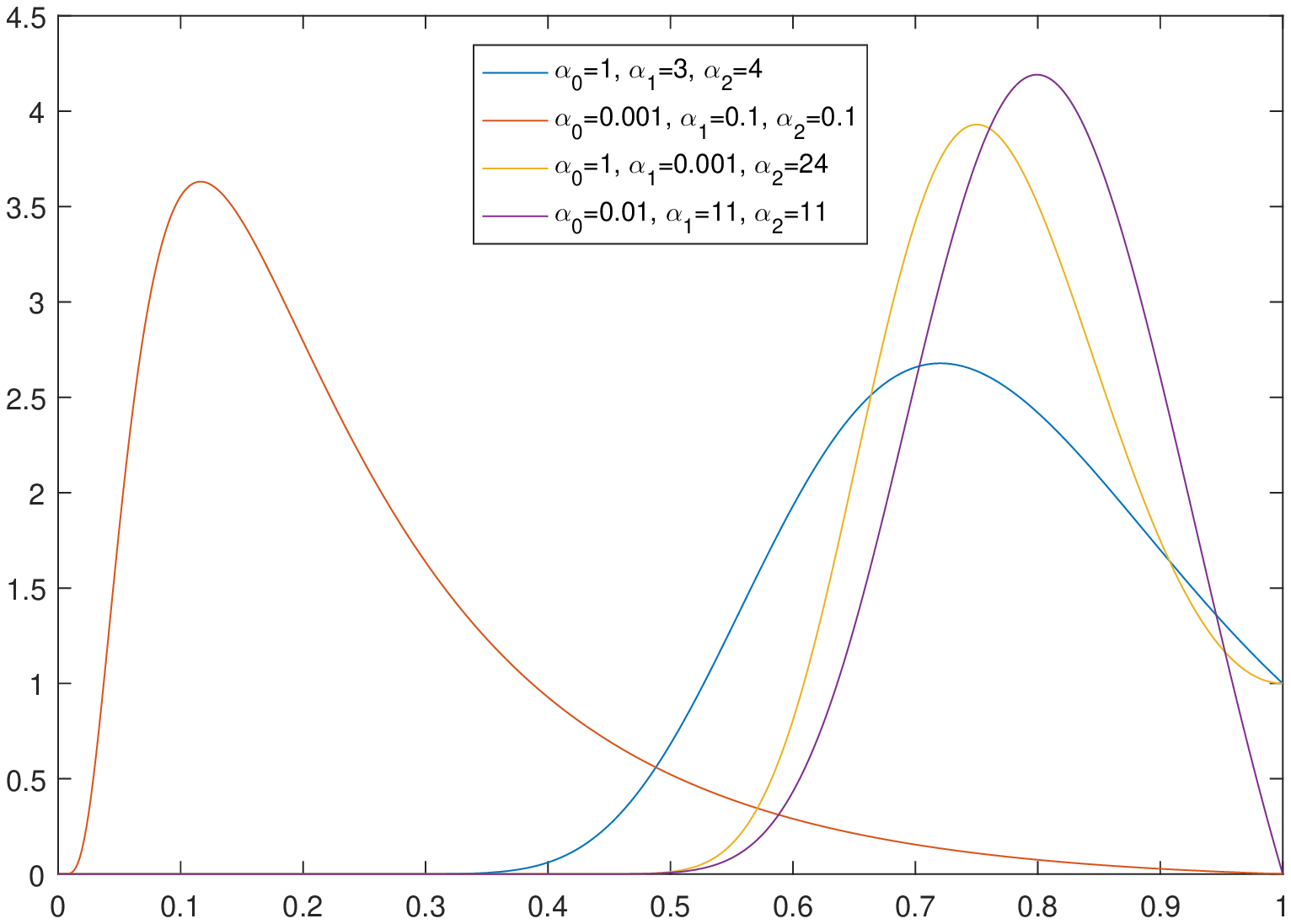}
    \caption{Shapes of the three-parameter case}
        \label{fig:fig6a}
  \end{subfigure}
  \begin{subfigure}[b]{0.5\linewidth}
    \centering
    \includegraphics[scale=0.42]{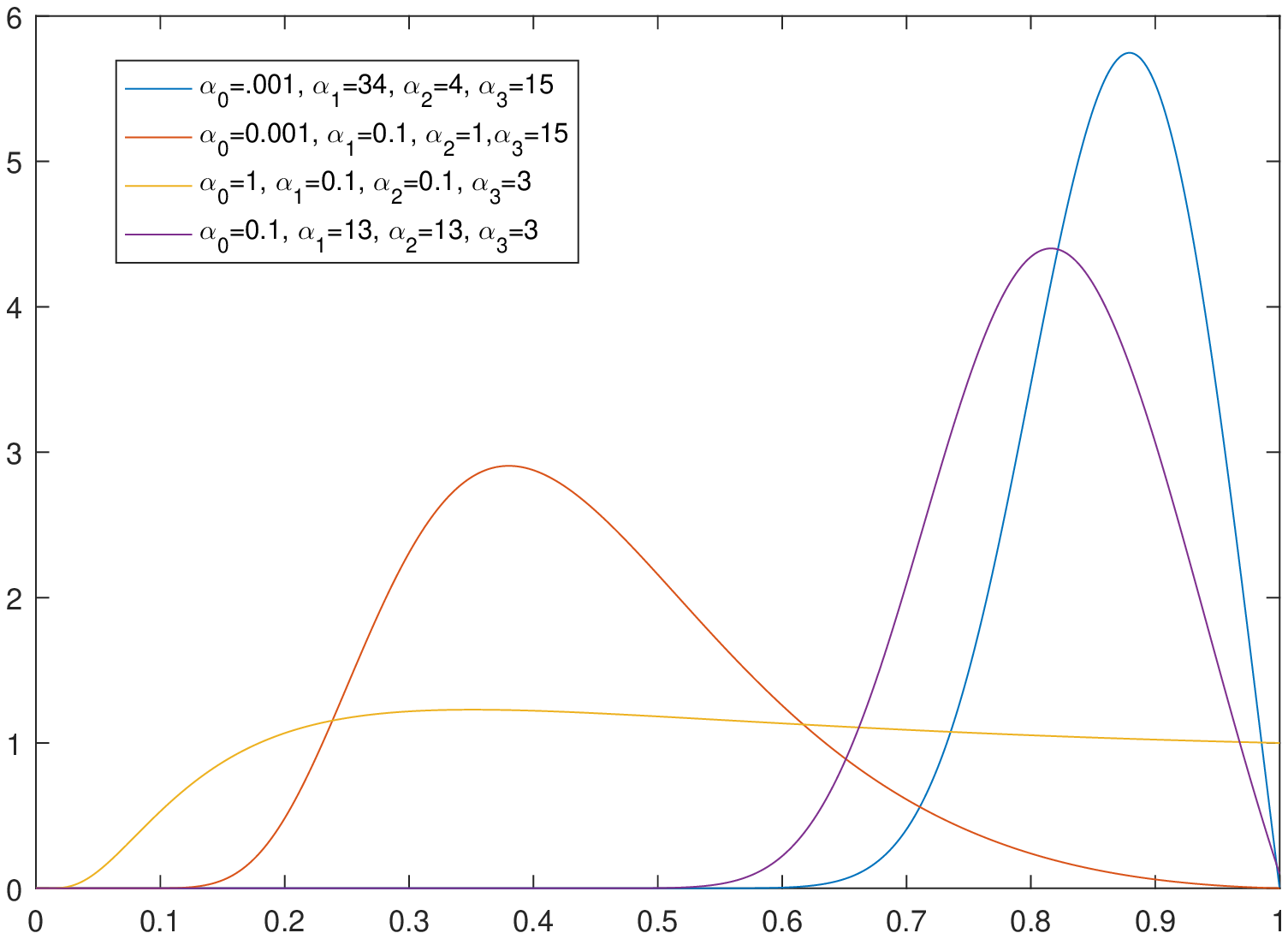}
    \caption{Shapes of the four-parameter cases}
          \label{fig:fig6b}
  \end{subfigure}
\caption{The EPD for the three and four-parameter cases}
      \label{fig:fig6}
  \end{figure}

This generalisation reduces to the two parameter extended power distribution if we set $\alpha_{2}=\alpha_{3}=\ldots=\alpha_{r-1}=0$ and to the $\text{Beta}(\alpha_{0}, 1)$, if we set $\alpha_{1}=\alpha_{2}=\alpha_{3}=\ldots=\alpha_{r-1}=0$. In a similar manner, setting $\alpha_{0}=1, \alpha_{1}=\alpha_{2}=\alpha_{3}=\ldots=\alpha_{r-1}=0$ gives the uniform distribution on $(0, 1)$. For this generalisation, the cumulative distribution function is
\begin{equation}\label{eqn:eqn23}
F(t)=\exp\left\{\sum_{h=1}^{r}(-1)^{h-1}\alpha_{h-1}(\log t)^{h}\right\}
\end{equation}
and we can simulate random variates using probability integral transform by finding the root of the polynomial
 \begin{equation*}
   \sum_{h=1}^{r}(-1)^{h-1}\alpha_{h-1}(\log T)^{h}-\log U=0
 \end{equation*}
which lies on $(0, 1)$. This is more complicated than in the two parameter case, however, there are a number of available mathematical programs for calculating roots of a polynomial and we can easily utilise the roots function in MATLAB for this purpose. A similar scenario applies when calculating the median of the generalised distribution and we obtain the median ($M$) as a root of the polynomial
\begin{equation*}
   \sum_{h=1}^{r}(-1)^{h-1}\alpha_{h-1}(\log M)^{h}-\log 0.5=0.
 \end{equation*}
Maximum likelihood estimators can be derived for the generalised extended power distribution by first obtaining the likelihood function and optimising using numerical methods.
The log-likelihood function for the generalised density is
\begin{align*}\label{eqn:eqn24}
\ell(\alpha_{0}, \alpha_{1}, \ldots, \alpha_{r-1})=&\sum_{i=1}^{n}\log\bigg(\sum_{h=1}^{r}\alpha_{h-1}(\log t_{i})^{h-1}h(-1)^{h-1}\bigg) -\sum_{i=1}^{n}\log t_{i}\\
&+\sum_{h=1}^{r}(-1)^{h-1}\alpha_{h-1}\sum_{i=1}^{n}(\log t_{i})^{h}.
\end{align*}
Differentiating the log-likelihood function with respect to $\alpha_{h-1}$, we get
\begin{equation}\label{eqn:eqn25}
\frac{\partial \ell}{\partial \alpha_{h-1}}=\sum_{i=1}^{n}\frac{(\log t_{i})^{h-1}h(-1)^{h-1}}{\sum_{h=1}^{r}\alpha_{h-1}(\log t_{i})^{h-1}h(-1)^{h-1}}+(-1)^{h-1}\sum_{i=1}^{n}(\log t_{i})^{h}
\end{equation}
and second derivative
\begin{equation}\label{eqn:eqn26}
\frac{\partial^{2} \ell}{\partial \alpha_{h-1}\alpha_{k-1}}=\sum_{i=1}^{n}\frac{(\log t_{i})^{h+k-2}hk(-1)^{h+k-1}}{\bigg(\sum_{h=1}^{r}\alpha_{h-1}(\log t_{i})^{h-1}h(-1)^{h-1}\bigg)^{2}}.
\end{equation}
\subsection{Applications and Examples}
In this subsection, we show examples where the extended power distribution is applicable. We also compare how well the extended power function fits actual data to how the Kumaraswamy distribution compares in performance. Finally, we will simulate data using the extended power distribution and try fitting the data with the Kumaraswamy distribution to see show cases where the extended power distribution has particular advantages over other bounded distributions. To fit the observed data, we will obtain maximum likelihood estimators for the parameters the distribution by maximising the log-likelihood functions for both the extended power distribution and the Kumaraswamy distribution. Table \ref{tab:table2} shows the Akaike information criterion (AIC) of the fitted distributions in each application. The corrected AIC (AICc) and the Bayesian information criterion (BIC) which penalises more for extra parameters are given in tables \ref{tab:table4} and \ref{tab:table5} respectively. The MLE for the fitted distributions are detailed in table \ref{tab:table3}. In most of the examples, the extended power distribution performed better than the Kumaraswamy distribution (using AIC, AICc and BIC). The only exception is example 2, where the Kumaraswamy had the least BIC. The difference between this BIC and the BIC of the 3 parameter EPD is quite negligible.
\begin{table}[h!]
\begin{tabular}{|p{2.3cm}|p{2.3cm}|p{2.3cm}|p{2.3cm}|p{2.3cm}|}
 \hline
Distribution  & Kumaras. &2-Par EPD &3-Par EPD & 4-Par EPD\\
 \hline
Example 1 & -88.6054 & -70.4173 & -91.9646 & $\mathbf{-92.8209}$ \\
\hline
Example 2 & -82.7203 & -67.5975 & $\mathbf{-84.6515}$ & -83.0609\\
\hline
Example 3 & -153.3079 & -170.5573 & $\mathbf{-177.6280}$ & -175.6280\\
\hline

Example 4 & -82.4570 & $\mathbf{-83.0053}$ & -81.0053 & -79.0053 \\
\hline
Example 5 &  -669.9358 & -719.9494 & -814.5105 & $\mathbf{-850.8778}$\\
\hline
Example 6 & -795.4145 & -965.8150 & $\mathbf{-996.5056}$ & $-994.5056$\\
\hline
Example 7 &  &$\mathbf{-302.7052}$ & -300.7052 & -298.7052 \\
\hline
\end{tabular}
\caption{AIC for fitted distributions in examples 1-7, with the best fitting distribution in bold}
\label{tab:table2}
\end{table}

\begin{table}[h!]
\begin{tabular}{|p{2.3cm}|p{2.3cm}|p{2.3cm}|p{2.3cm}|p{2.3cm}|}
 \hline
Distribution  & Kumaras. &2-Par EPD &3-Par EPD & 4-Par EPD\\
 \hline
Example 1 & -88.4054 & -70.2173 & -91.5579 & $\mathbf{-92.1313}$ \\
\hline
Example 2 & -82.5203 & -67.3975 & $\mathbf{-84.3447}$ & -82.3712\\
\hline
Example 3 & -153.1364 & -170.3859 & $\mathbf{-177.2802}$ & -175.0398\\
\hline

Example 4 & -82.2570 & $\mathbf{-82.8053}$ & -80.5985 & -78.3156 \\
\hline
Example 5 &  -669.9052 & -719.9188 & -814.4491 & $\mathbf{-850.7753}$\\
\hline
Example 6 & -795.4024 & -965.8030 & $\mathbf{-996.4815}$ & $-994.4654$\\
\hline
Example 7 &  &$\mathbf{-302.6230}$ & -300.5396 & -298.4274 \\
\hline
\end{tabular}
\caption{AICc for fitted distributions in examples 1-7, with the best fitting distribution in bold}
\label{tab:table4}
\end{table}

\begin{table}[h!]
\begin{tabular}{|p{2.3cm}|p{2.3cm}|p{2.3cm}|p{2.3cm}|p{2.3cm}|}
 \hline
Distribution  & Kumaras. &2-Par EPD &3-Par EPD & 4-Par EPD\\
 \hline
Example 1 & -84.3191 & -66.1311 & $\mathbf{-85.5352}$ & -84.2484 \\
\hline
Example 2 & $\mathbf{-78.4340}$ & -63.3113 & -78.2221 & -74.4883\\
\hline
Example 3 & -148.7269 & -165.9764 & $\mathbf{-170.7566}$ & -166.4662\\
\hline

Example 4 & -78.1708 & $\mathbf{-78.7190}$ & -74.5759 & -70.4327 \\
\hline
Example 5 &  -661.9780 & -711.9916 & -802.5738 & $\mathbf{-834.9623}$\\
\hline
Example 6 & -785.5990 & -955.9995 & $\mathbf{-981.7823}$ & $-974.8745$\\
\hline
Example 7 &  &$\mathbf{-296.6973}$ & -291.6933 & -286.6894 \\
\hline
\end{tabular}
\caption{BIC for fitted distributions in examples 1-7, with the best fitting distribution in bold}
\label{tab:table5}
\end{table}

\begin{table}[h!]
\begin{tabular}{|p{2cm}|p{2.3cm}|p{2.3cm}|p{2.3cm}|p{2.3cm}|}
 \hline
Distribution  & Kumaras. &2-Par EPD &3-Par EPD & 4-Par EPD\\
 \hline
Example 1 & $\hat{\alpha}=4.51, \hat{\beta}=15.21$& $\hat{\alpha}_{0}=0.00,   \hat{\alpha}_{1}= 1.72$ & $\hat{\alpha}_{0}=0.00,   \hat{\alpha}_{1}= 0.00, \hat{\alpha}_{2}=2.00$ & $\hat{\alpha}_{0}=0.00,   \alpha_{1}=0.00,  \hat{\alpha}_{2}= 0.73, \hat{\alpha}_{3}=1.39$ \\
\hline
Example 2 & $\hat{\alpha}=4.31, \hat{\beta}=13.06$& $\hat{\alpha}_{0}=0.00,   \hat{\alpha}_{1}= 1.68$ & $\hat{\alpha}_{0}=0.00,   \hat{\alpha}_{1}= 0.00, \hat{\alpha}_{2}=1.90$ & $\hat{\alpha}_{0}=0.00,   \alpha_{1}=0.00,  \hat{\alpha}_{2}= 1.45, \hat{\alpha}_{3}=0.47$\\
\hline

Example 3 & $\hat{\alpha}=0.66, \hat{\beta}=3.44$& $\hat{\alpha}_{0}=0.01,   \hat{\alpha}_{1}=0.10$ & $\hat{\alpha}_{0}=0.04,   \hat{\alpha}_{1}= 0.00, \hat{\alpha}_{2}=0.02$ & $\hat{\alpha}_{0}=0.04,   \hat{\alpha}_{1}=0.00, \hat{\alpha}_{2}=0.02, \hat{\alpha}_{3}=0.00$\\
\hline

Example 4 & $\hat{\alpha}=5.85, \hat{\beta}=3.03$& $\hat{\alpha}_{0}=0.25, \hat{\alpha}_{1}=7.03$ & $\hat{\alpha}_{0}=0.25, \hat{\alpha}_{1}= 7.03, \hat{\alpha}_{2}=0.00$ & $\hat{\alpha}_{0}=0.25, \hat{\alpha}_{1}= 7.03, \hat{\alpha}_{2}=0.00, \hat{\alpha}_{3}=0.00$ \\
\hline
Example 5 &  $\hat{\alpha}=1.00, \hat{\beta}=5.28$& $\hat{\alpha}_{0}=0.00, \hat{\alpha}_{1}=0.17$ & $\hat{\alpha}_{0}=0.03, \hat{\alpha}_{1}= 0.00, \hat{\alpha}_{2}=0.06$ & $\hat{\alpha}_{0}=0.07, \hat{\alpha}_{1}= 0.00, \hat{\alpha}_{2}=0.00, \hat{\alpha}_{3}=0.02$\\
\hline

Example 6 &  $\hat{\alpha}=3.58, \hat{\beta}= 2.05$& $\hat{\alpha}_{0}=0.51, \hat{\alpha}_{1}=3.38$ & $\hat{\alpha}_{0}=0.90, \hat{\alpha}_{1}= 0.62, \hat{\alpha}_{2}=3.22$ & $\hat{\alpha}_{0}=0.90, \hat{\alpha}_{1}= 0.62, \hat{\alpha}_{2}=3.22, \hat{\alpha}_{3}=0.00$\\
\hline

Example 7 & & $\hat{\alpha}_{0}=6.53, \hat{\alpha}_{1}=0.00$ & $\hat{\alpha}_{0}=6.53, \hat{\alpha}_{1}=0.00, \hat{\alpha}_{2}=0.00$ & $\hat{\alpha}_{0}=6.53, \hat{\alpha}_{1}= 0.00, \hat{\alpha}_{2}=0.00, \hat{\alpha}_{3}=0.00$\\
\hline
\end{tabular}
\caption{MLE for fitted distributions in examples 1-7}
\label{tab:table3}
\end{table}

\subsubsection{Example 1}
In this example, we explore data on US senate voting patterns from 1953 to 2015. The observed data are proportion of party unity votes in which a majority of voting Democrats opposed a majority of voting Republicans. The data is taken from the Brookings Institution (\cite{Brookings2017}) and is available for free.
A histogram of the actual data and its empirical cumulative distribution along with fitted density and distribution functions are given in figure \ref{fig:fig3}. From the fitted distributions, we see that the four-parameter extended power distribution best fits the actual data.
\begin{figure}[ht]
\centering
    \begin{subfigure}[b]{0.5\linewidth}
    \centering
    \includegraphics[scale=0.4]{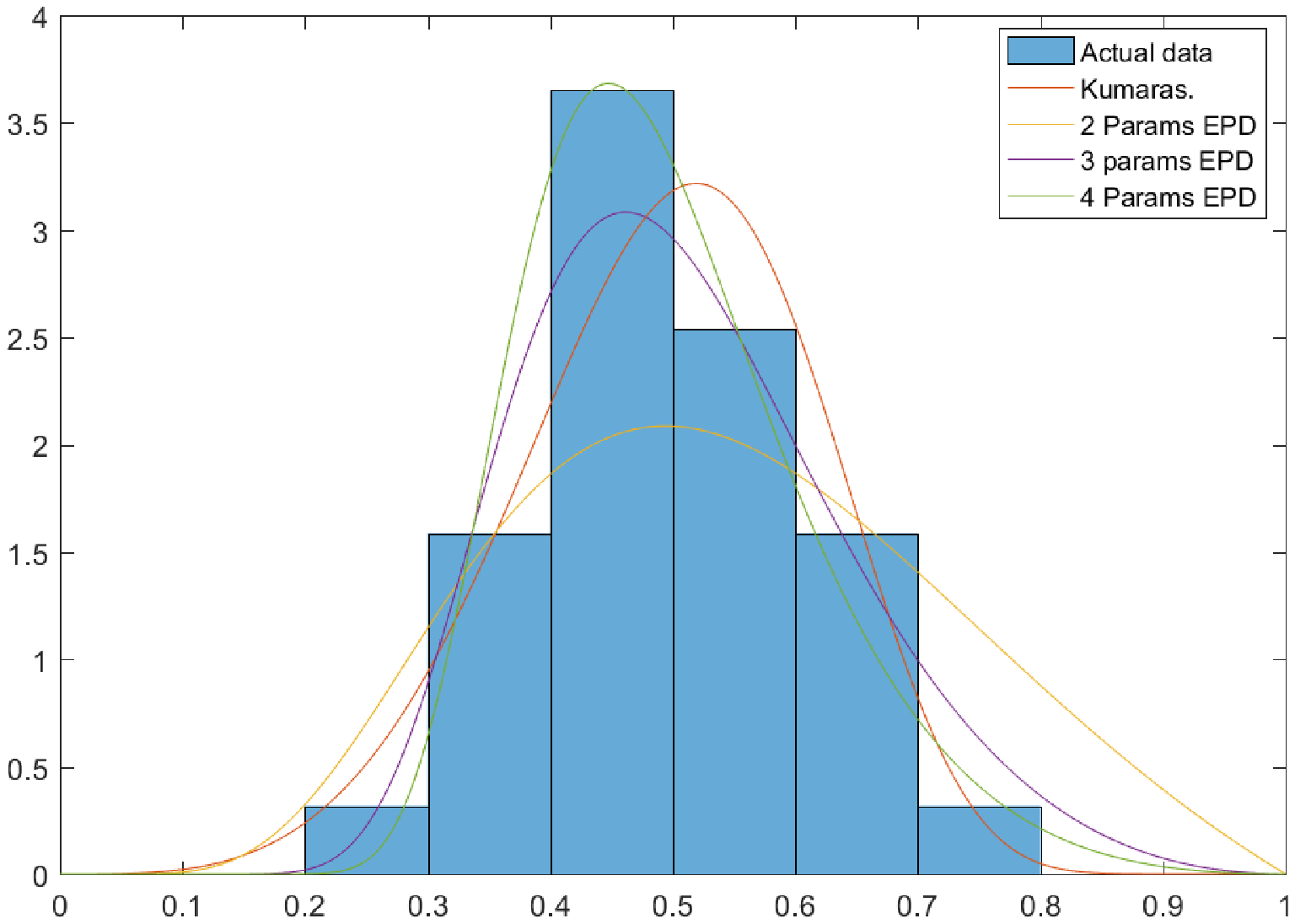}
    \caption{Density functions fitted to histogram}
        \label{fig:fig3a}
  \end{subfigure}
  \begin{subfigure}[b]{0.5\linewidth}
    \centering
    \includegraphics[scale=0.4]{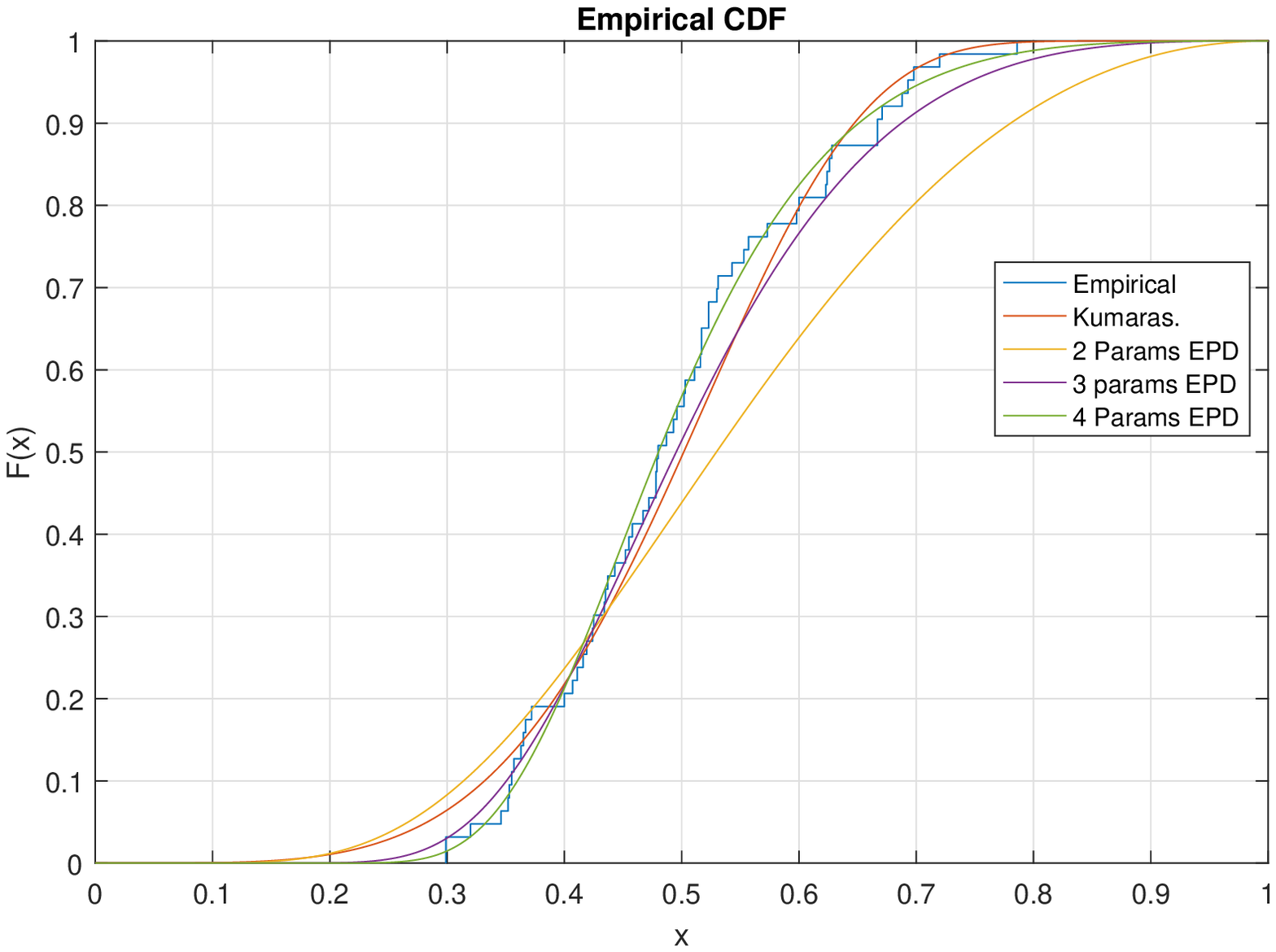}
    \caption{Fitting cdfs to the empirical cdf}
          \label{fig:fig3b}
  \end{subfigure}
\caption{Fitting proportion of US Senate party unity votes using the Kumaraswamy and extended power distributions}
      \label{fig:fig3}
  \end{figure}

\subsubsection{Example 2}
This second data explores proportion of unity votes in the US House of Representatives from 1953 to 2015. The data is also from \cite{Brookings2017}. A histogram (and empirical cumulative distribution) of the actual data and fitted distributions using MLE are shown in figure \ref{fig:fig4}. In this case, the three-parameter and four-parameter extended power distribution give the best fit, with the four-parameter case having a slightly higher peak. From the cumulative distribution plot, we see that at the lower tail, the three and four-parameter extended power distribution function best fits the empirical distribution function.
\begin{figure}[ht]
\centering
    \begin{subfigure}[b]{0.5\linewidth}
    \centering
    \includegraphics[scale=0.4]{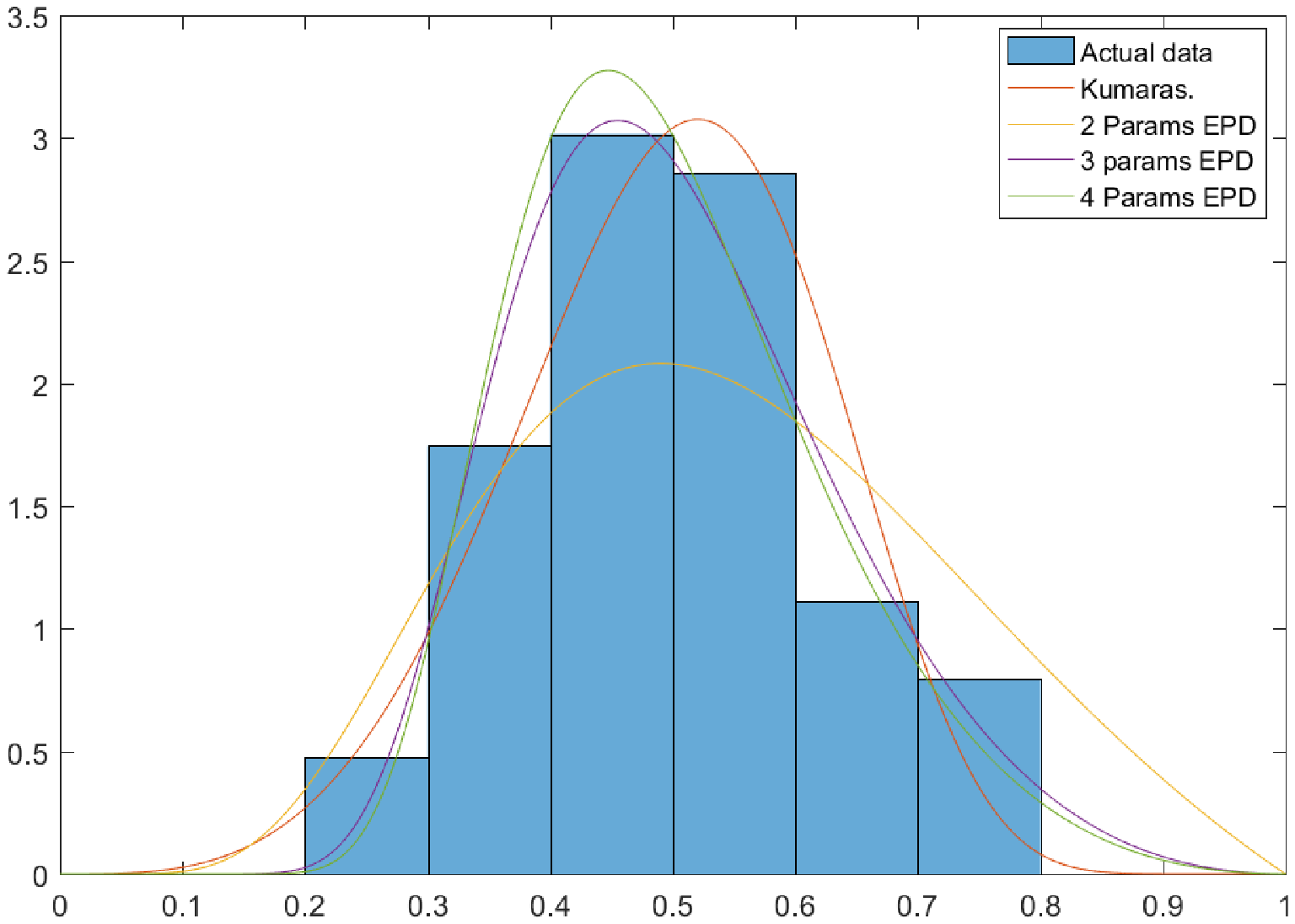}
    \caption{Density functions fitted to histogram}
        \label{fig:fig4a}
  \end{subfigure}
  \begin{subfigure}[b]{0.5\linewidth}
    \centering
    \includegraphics[scale=0.4]{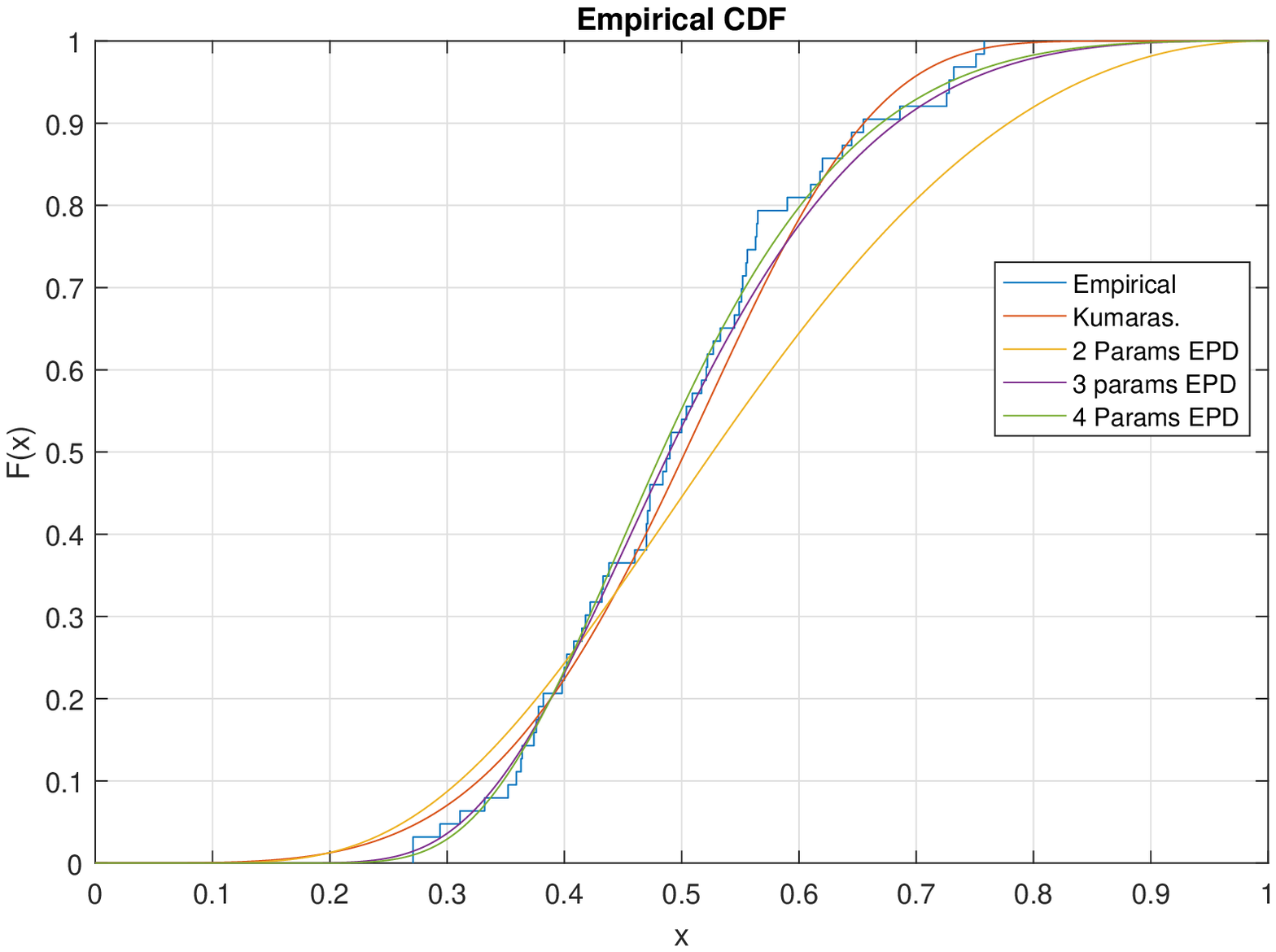}
    \caption{Fitting cdfs to the empirical cdf}
          \label{fig:fig4b}
  \end{subfigure}
\caption{Fitting proportion of US House of Representatives party unity votes using the Kumaraswamy and extended power distributions}
      \label{fig:fig4}
  \end{figure}

\subsubsection{Example 3}
This data which is taken  from \cite{jodra2016note} and is a measure of a firm's risk management cost effectiveness given as a proportion. This data has been used in \cite{jodra2016note} for regression modelling using the beta and log-Lindley distributions. The extended power distribution gives the best fit for this data and even has smaller AIC than that obtained for the log-Lindley distribution (\cite{gomez2014log}) which is $-149.2083$.
\begin{figure}[ht]
\centering
    \begin{subfigure}[b]{0.5\linewidth}
    \centering
    \includegraphics[scale=0.4]{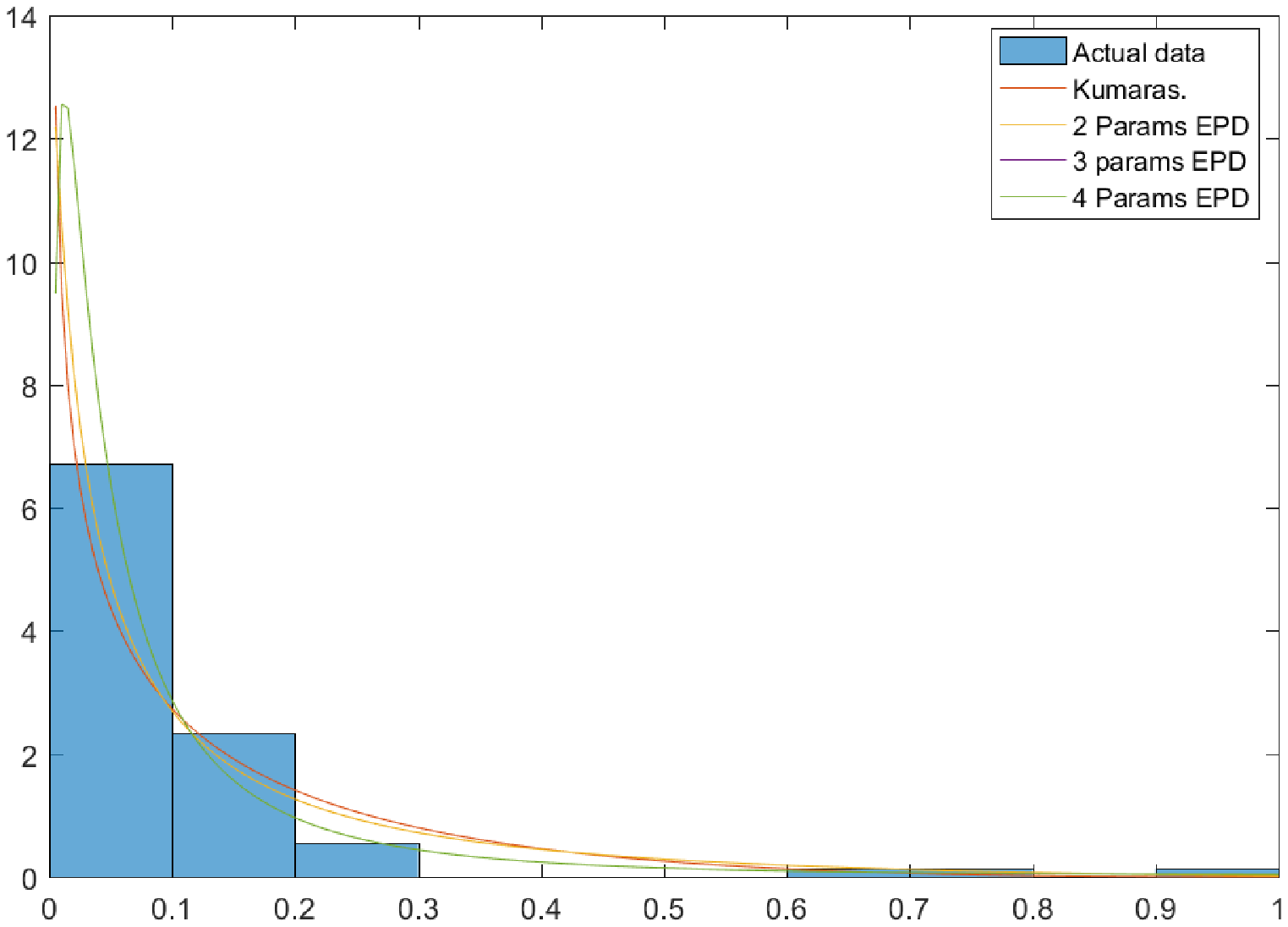}
    \caption{Density functions fitted to histogram}
        \label{fig:fig13a}
  \end{subfigure}
  \begin{subfigure}[b]{0.5\linewidth}
    \centering
    \includegraphics[scale=0.4]{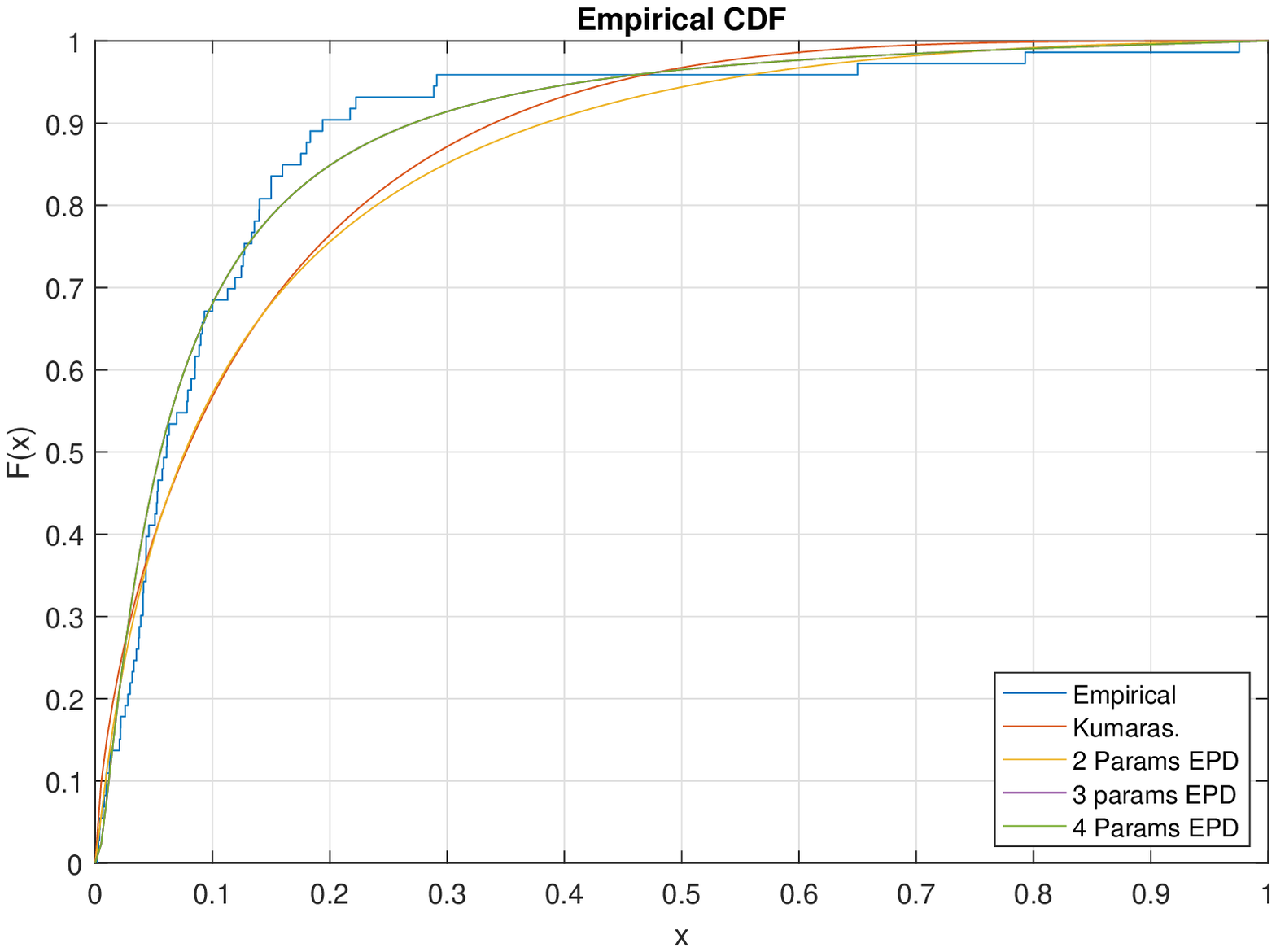}
    \caption{Fitting cdfs to the empirical cdf}
          \label{fig:fig13b}
  \end{subfigure}
\caption{Fitting risk management cost effectiveness using the Kumaraswamy and extended power distributions}
      \label{fig:fig13}
  \end{figure}

\subsubsection{Example 4}
This example involves proportion of presidential bill victories in the US Senate from President Einsenhower in 1953 to President Obama in 2015 (obtained from \cite{Brookings2017}). The fitted distributions are shown in figure \ref{fig:fig9}. In this case, the two-parameter extended power distribution seems to give the best fit.
\begin{figure}[ht]
\centering
    \begin{subfigure}[b]{0.5\linewidth}
    \centering
    \includegraphics[scale=0.4]{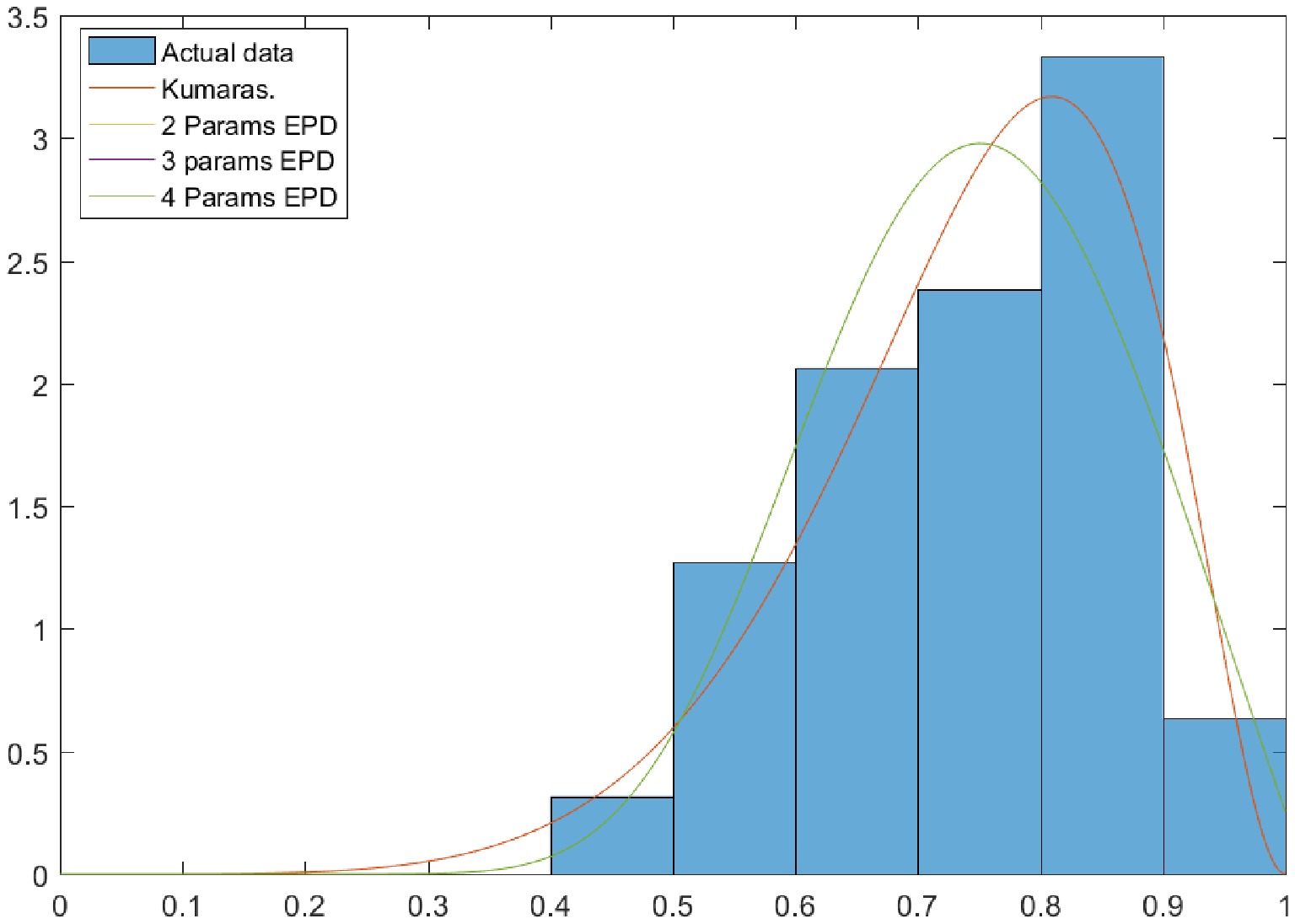}
    \caption{Density functions fitted to histogram}
        \label{fig:fig9a}
  \end{subfigure}
  \begin{subfigure}[b]{0.5\linewidth}
    \centering
    \includegraphics[scale=0.4]{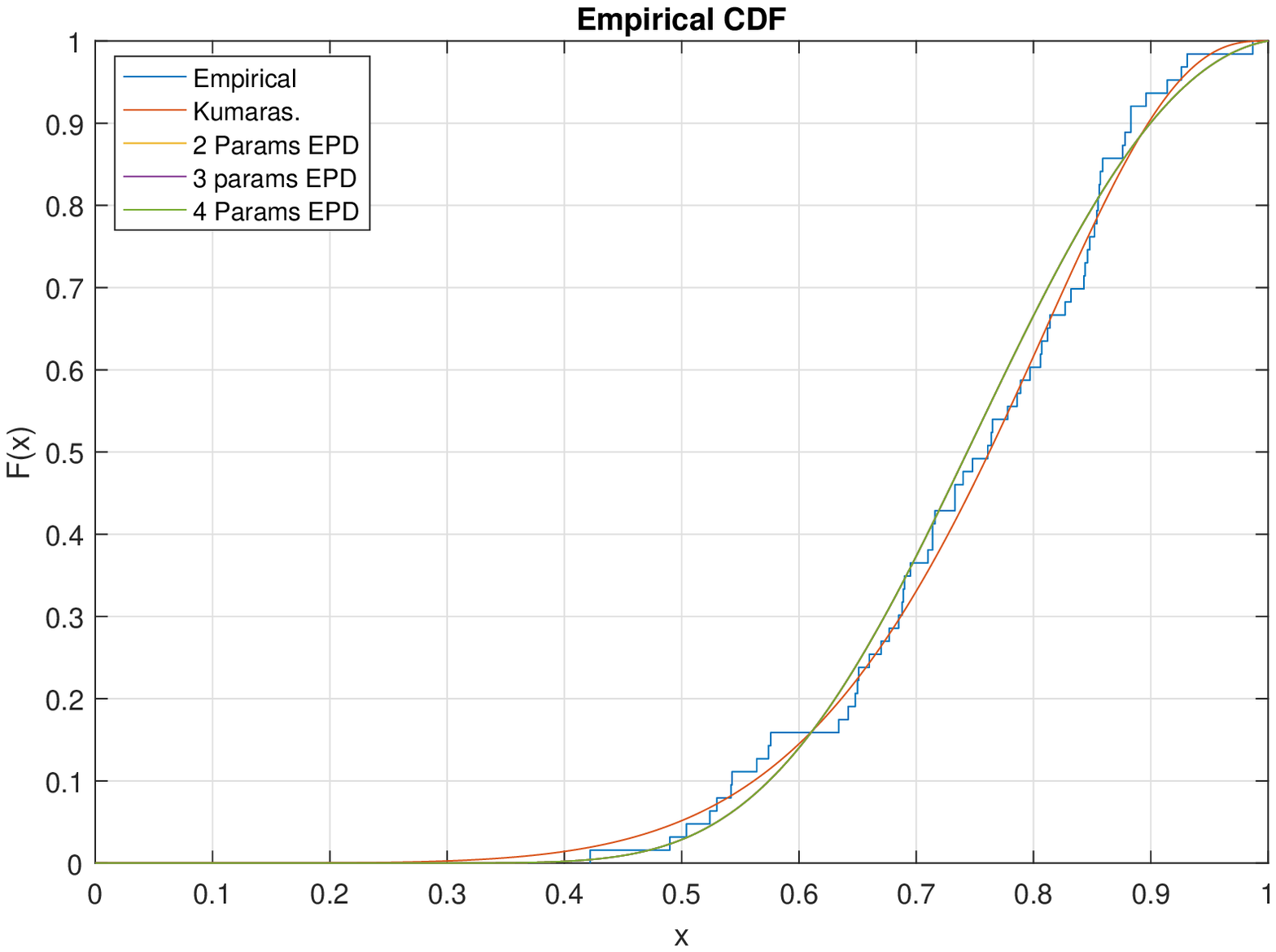}
    \caption{Fitting cdfs to the empirical cdf}
          \label{fig:fig9b}
  \end{subfigure}
\caption{Fitting proportion of US senators of who support presidential bills using the Kumaraswamy and extended power distributions}
      \label{fig:fig9}
  \end{figure}

\subsubsection{Example 5}
This example is taken from the 2011 census data on the proportion of minority ethnic groups across different local authorities in England and Wales. The data is obtained from the Office of National Statistics 2011 census (\cite{ONSrace}) and the minorities  include groups such as white Irish, white and black Caribbean, Arabian, etc. The extended power distribution has the best fit for this example.
\begin{figure}[ht]
    \begin{subfigure}[b]{0.5\linewidth}
    \centering
    \includegraphics[scale=0.4]{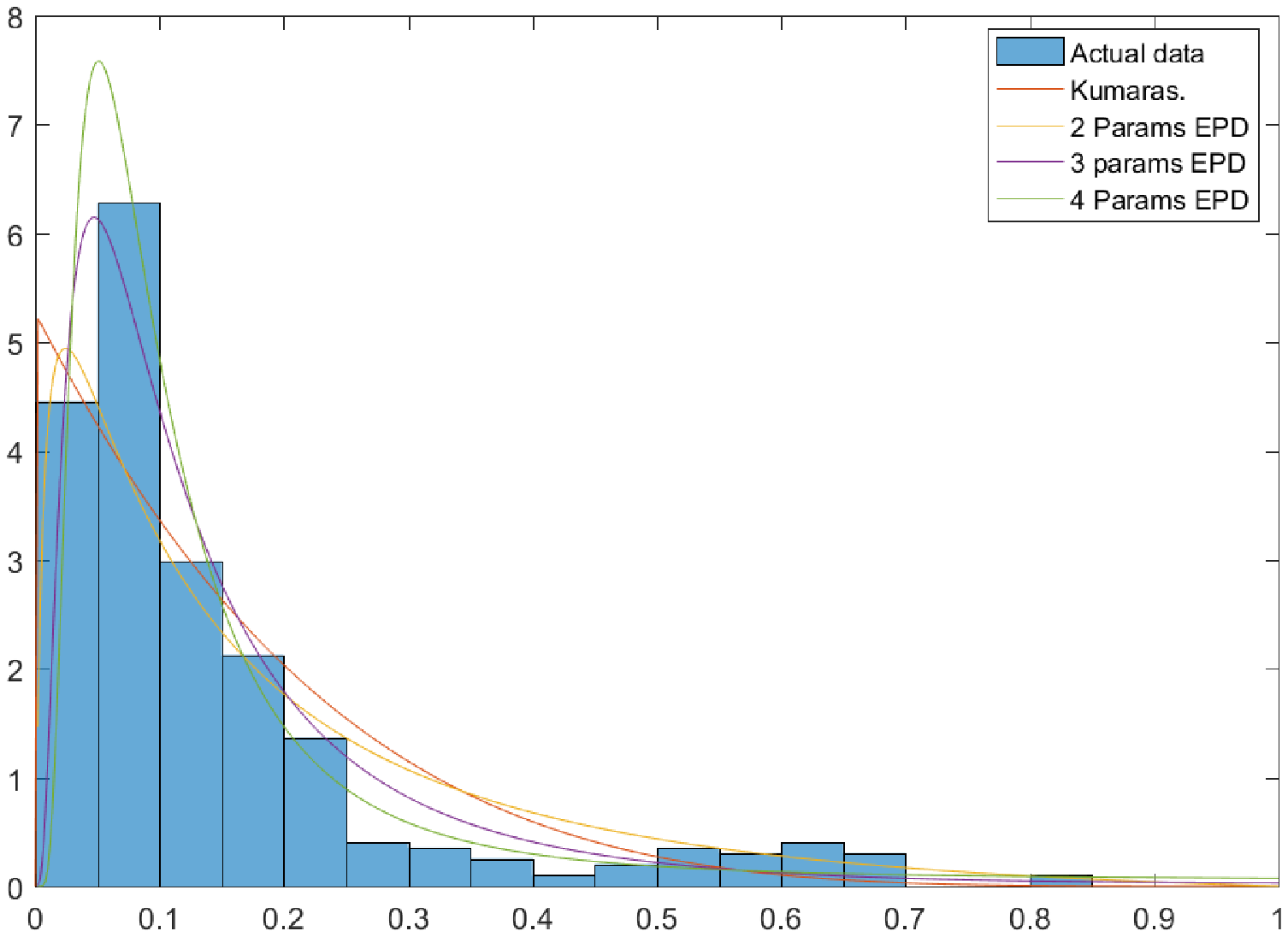}
    \caption{Density functions fitted to histogram}
        \label{fig:fig10a}
  \end{subfigure}
  \begin{subfigure}[b]{0.5\linewidth}
    \centering
    \includegraphics[scale=0.4]{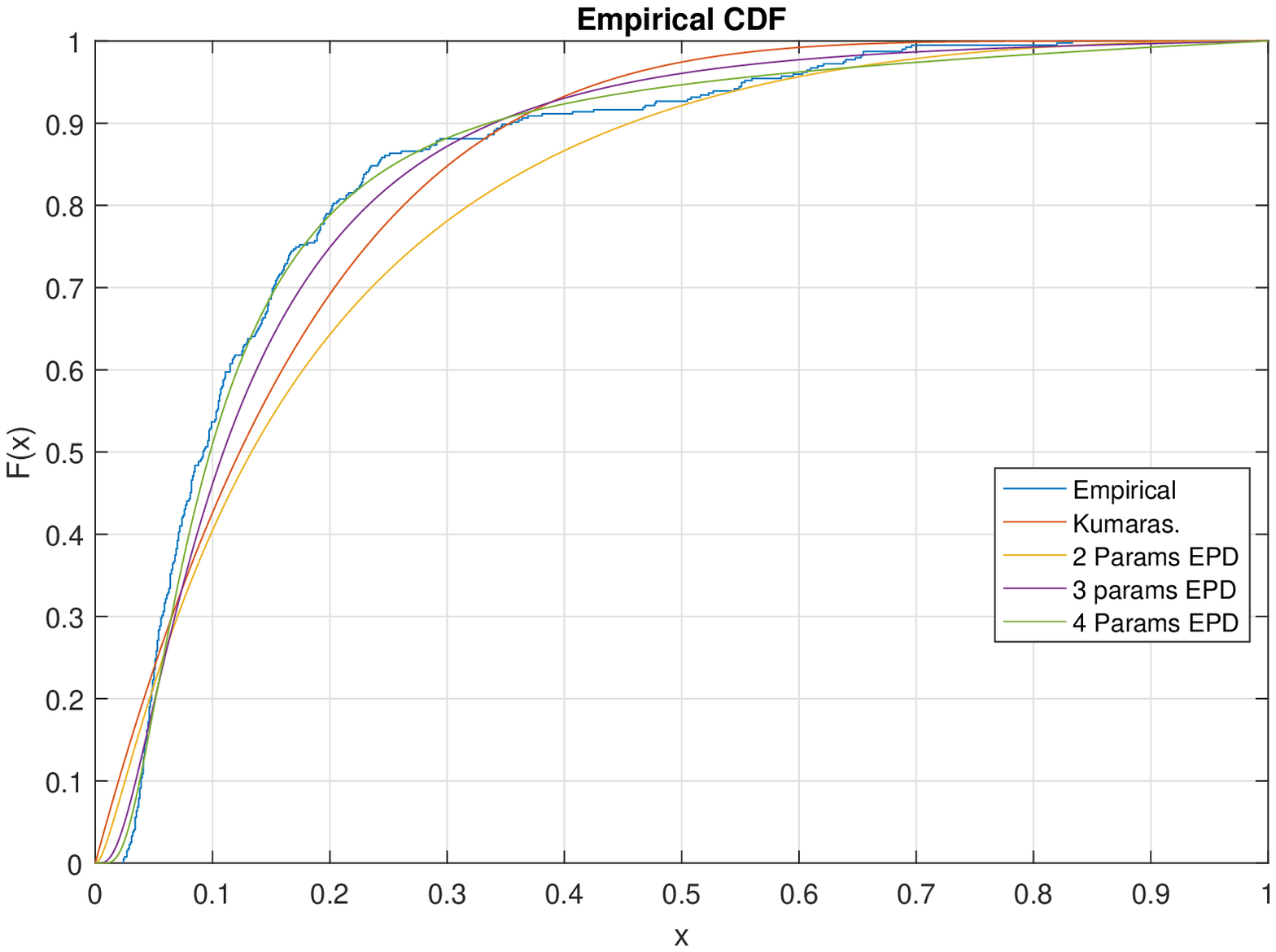}
    \caption{Fitting cdfs to the empirical cdf}
          \label{fig:fig10b}
  \end{subfigure}
\caption{Fitting proportion of ethnic minorities in different local authorities in England and Wales using the Kumaraswamy and extended power distributions}
      \label{fig:fig10}
  \end{figure}

\subsubsection{Example 6}
In this example, we simulate data ($n=1000$) from the three-parameter extended power distribution with parameters $\alpha_{0}=1$, $\alpha_{1}=0.001$ and $\alpha_{4}=4$. We then fit the simulated data using the Kumaraswamy distribution. From figure \ref{fig:fig7}, we see that the the Kumaraswamy distribution fails to properly fit the data as values of the random variable approaches 1.
\begin{figure}[ht]
\centering
    \begin{subfigure}[b]{0.5\linewidth}
    \centering
    \includegraphics[scale=0.4]{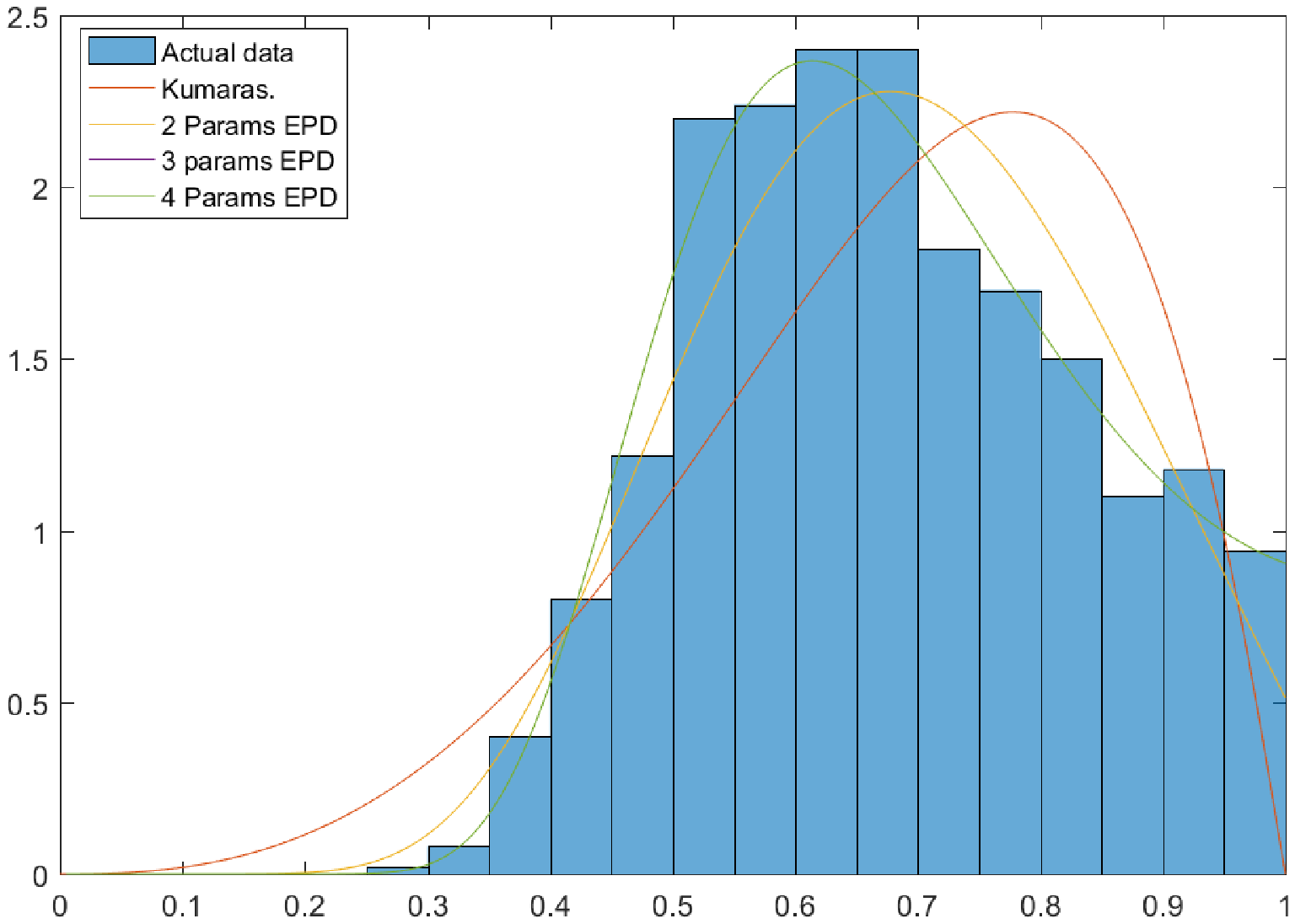}
    \caption{Density functions fitted to histogram}
        \label{fig:fig7a}
  \end{subfigure}
  \begin{subfigure}[b]{0.5\linewidth}
    \centering
    \includegraphics[scale=0.4]{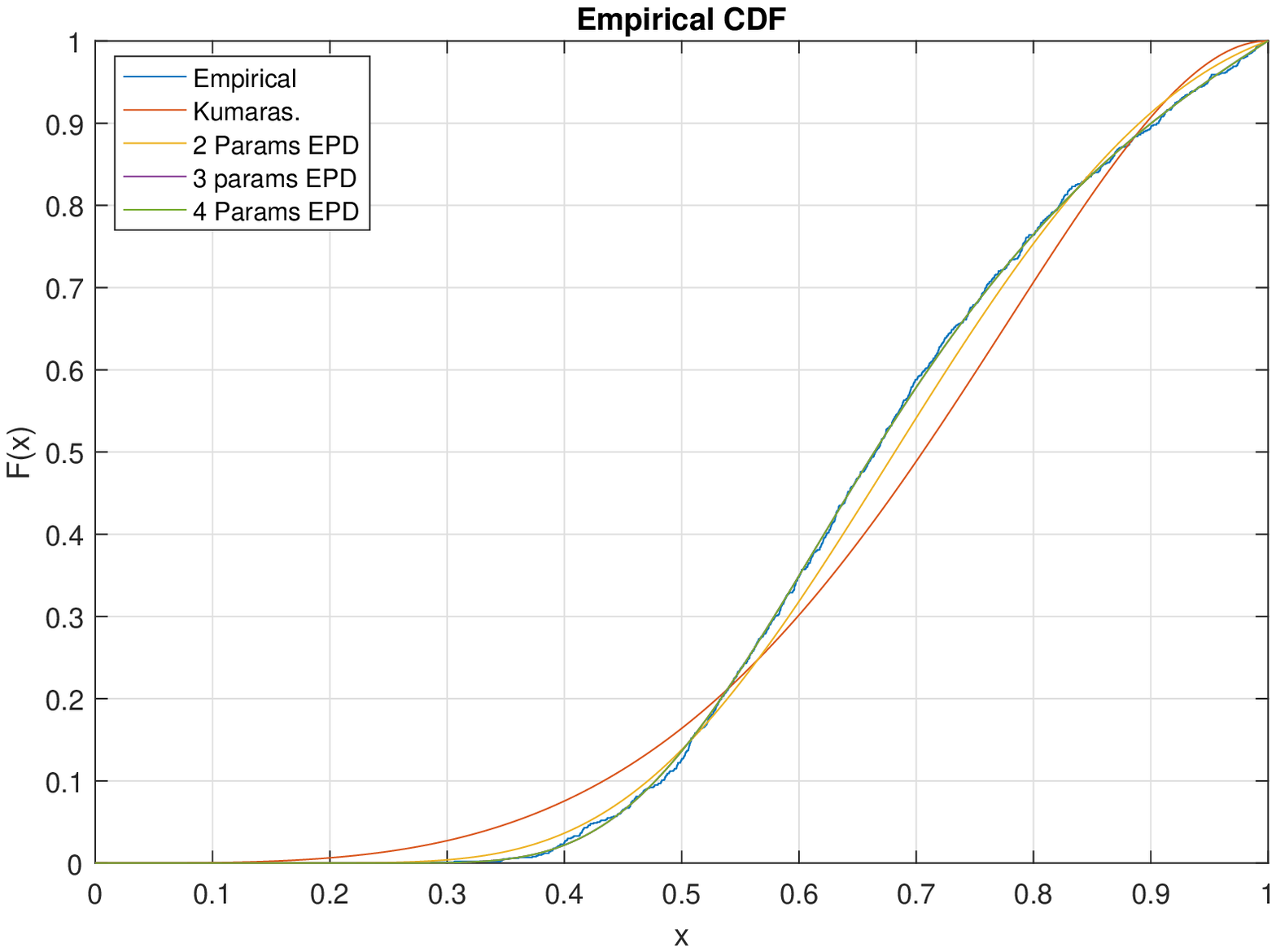}
    \caption{Fitting cdfs to the empirical cdf}
          \label{fig:fig7b}
  \end{subfigure}
\caption{Fitting a Kumaraswamy distribution and extended power distribution to data simulated from the three-parameter extended power distribution}
\label{fig:fig7}
\end{figure}

\subsubsection{Example 7}
This example, obtained from the UNICEF website (\cite{UNICEFeducation}) contains the proportion of literate youths in 149 different countries (updated in October, 2015). Figure \ref{fig:fig11} show a plot of the actual data as well as the maximum likelihood fits. In this example, the Kumaraswamy distribution is not applicable, because some countries have youth literacy of $100\%$ (that is proportion is 1), however, the Kumaraswamy distribution has density converging to 0 (for $\beta>1$) and to $\infty$ (for $\beta<1$) when $t$ is exactly 1, making the log-likelihood function undefined at $t=1$ and the MLE cannot be computed. For this example, the two-parameter extended power distribution has the smallest AIC values, hence gives the best fit.
\begin{figure}[ht]
\centering
    \begin{subfigure}[b]{0.5\linewidth}
    \centering
    \includegraphics[scale=0.4]{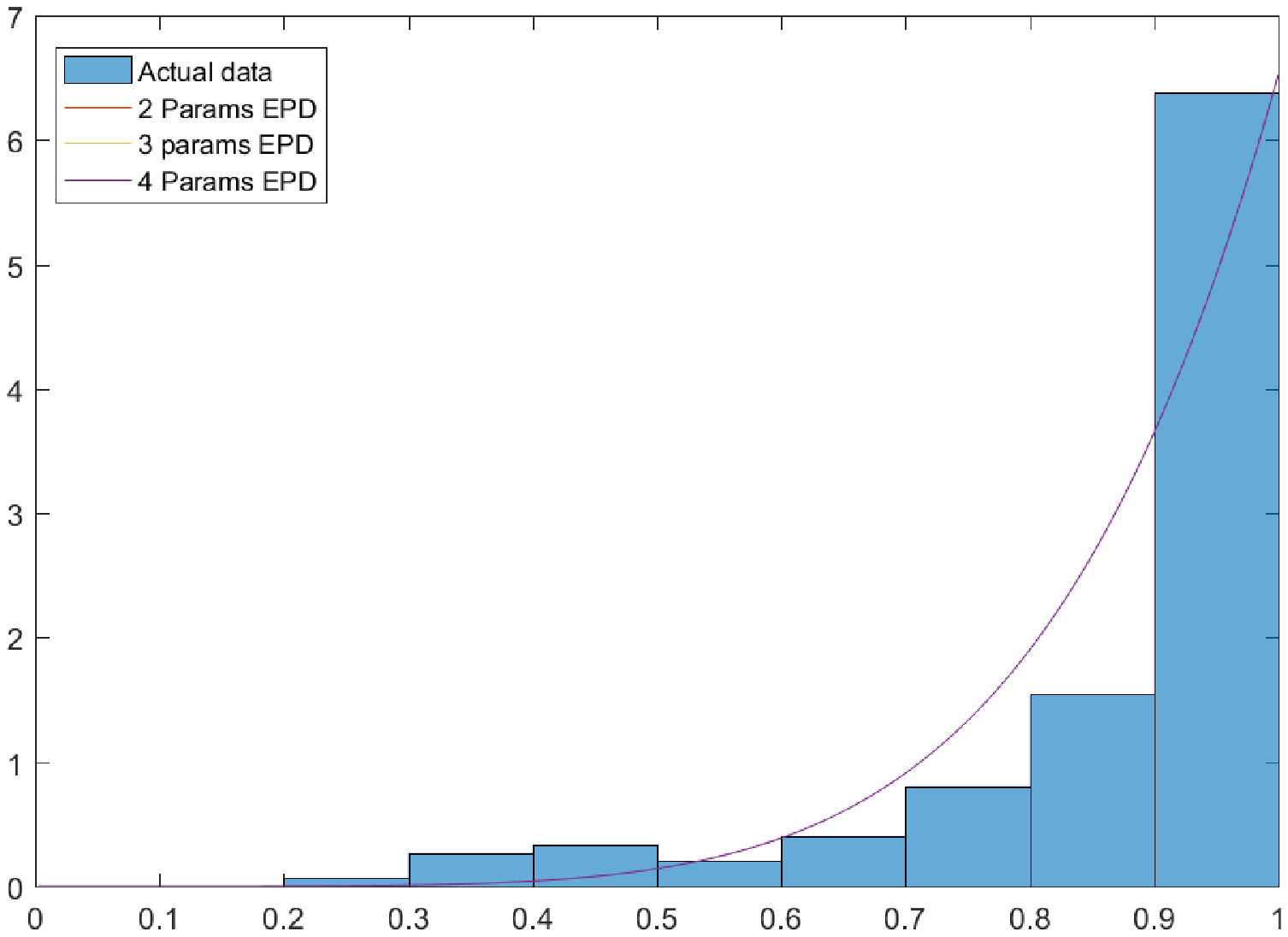}
    \caption{Density functions fitted to histogram}
        \label{fig:fig11a}
  \end{subfigure}
  \begin{subfigure}[b]{0.5\linewidth}
    \centering
    \includegraphics[scale=0.4]{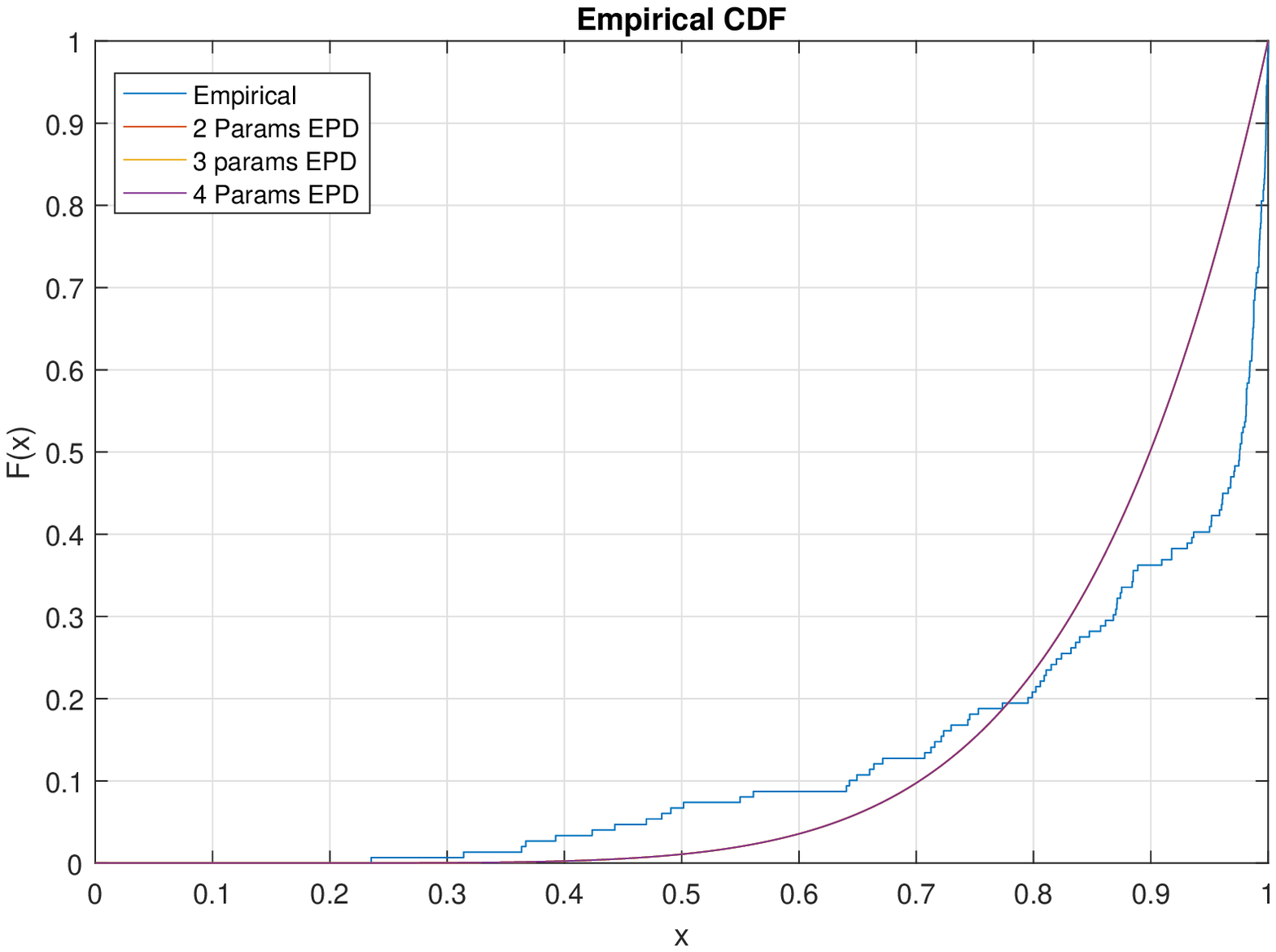}
    \caption{Fitting cdfs to the empirical cdf}
          \label{fig:fig11b}
  \end{subfigure}
\caption{Fitting youths literacy rates in 149 different countries using the Kumaraswamy and extended power distributions}
      \label{fig:fig11}
  \end{figure}
\section{Complementary Distribution}
The complementary distribution for the beta distribution has been defined in \cite{jones2002complementary} and are obtained by considering the quantile functions and as the cumulative distribution function of a probability distribution on $(0, 1)$. This same procedure can be considered for other bounded distribution on $(0, 1)$ to get a new bounded distribution. In the complementary beta distribution, the form of density function and moments involve the use of special function as given in \cite{jones2002complementary}. However, for the complementary Kumarawamy distribution (\cite{jones2009kumaraswamy}), even though the density function can be written in a nice form, there is nothing new to be seen in its form because it is equivalent to  $g(1-t, \frac{1}{\beta}, \frac{1}{\alpha})$, where $g\sim \text{Kumaraswamy}(t, \alpha, \beta)$.

In this section, we will define the complementary extended power distribution, its properties such as moments, quantiles and parameter estimation. The density function of the complementary extended power distribution is obtained by taking the first derivative of the quantile function. The quantile function is given as
\begin{equation*}
Q(t)=F^{-1}(t)=\exp\left\{\frac{-\alpha_{0}+(\alpha_{0}^{2}-4\alpha_{1}\log(t))^{1/2}}{-2\alpha_{1}}\right\},
\end{equation*}
hence the density function of the proposed complementary extended power distribution is
\begin{equation}\label{eqn:eqn16}
q(t)=\frac{1}{t}(\alpha_{0}^{2}-4\alpha_{1}\log t)^{-1/2}\exp\left\{\frac{-\alpha_{0}+(\alpha_{0}^{2}-4\alpha_{1}\log(t))^{1/2}}{-2\alpha_{1}}\right\}.
\end{equation}
An interesting properties of this distribution on first sight is that both its density function and distribution are available in a simple closed form without the need for any special function. Secondly, it has a form which is not exactly similar to the extended power distribution, implying some new information may be gained.

An interesting special case of the complementary extended power distribution is obtained by setting $\alpha_{1}=0$. Simply replacing $\alpha_{1}=0$ in equation \ref{eqn:eqn16}, will make the exponent indeterminate. To deal with this problem, we use binomial expansion on $(\alpha_{0}^{2}-4\alpha_{1}\log(t))^{1/2}$, hence the density function becomes
\begin{equation*}
q(t)=  \frac{1}{t}(\alpha_{0}^{2}-4\alpha_{1}\log t)^{-1/2}\exp\left\{ \frac{1}{\alpha_{0}}\log t +\frac{\alpha_{1}(\log t)^{2}}{\alpha_{0}^{3}}+\frac{2\alpha_{1}^{2}(\log t)^{3}}{\alpha_{0}^{5}}+ \ldots\right\}, \quad t \in (0, 1)
\end{equation*}
and setting $\alpha_{1}=0$ gives
\begin{equation*}
q(t)=\frac{1}{\alpha_{0}}t^{\frac{1}{\alpha_{0}}-1}
\end{equation*}
which corresponds to a $\text{Beta}(\frac{1}{\alpha_{0}}, 1)$ or a $\text{Kumaraswamy}(\frac{1}{\alpha_{0}}, 1)$. Like in the extended power distribution considered earlier, fixing $\alpha_{0}=1$ and $\alpha_{1}=0$, $q(t)$ reduces to the density for a uniform distribution on $(0, 1)$. If T is a random variable from the two-parameter complementary extended power distribution, with $\alpha_{1}=0$, then the random variable $V=-\log T$ follows an exponential distribution with density function
\begin{equation*}
  f(v)=\frac{1}{\alpha_{0}}\exp\{-\frac{v}{\alpha_{0}}\}.
\end{equation*}

Generating random variates from the complementary extended power distribution is straightforward using the probability integral transform. If $U\sim U(0, 1)$, then we simulate random variables T using,
\begin{equation}\label{eqn:eqn19}
T=\exp\left\{ \alpha_{0}\log U-\alpha_{1}(\log U)^{2}\right\}.
\end{equation}
\subsection{Moments, Quantiles and Mode}
In this section, we will give formulas for the moments, quantiles and mode of the complementary extended power distribution.
We will begin by giving a formula for the $kth$ moment of the complementary extended power distribution. This makes it easier to calculate quantities such as the variance, skewness and kurtosis.
\begin{thm}
The $kth$  moment of the complementary extended power distribution is
\begin{equation}\label{eqn:eqn18}
E(T^{k})=\frac{1}{2}\sqrt{\frac{\pi}{\alpha_{1}k}}\exp\left\{\frac{(4\alpha_{0}\alpha_{1}k+4\alpha_{1})^{2}}{64\alpha_{1}^{3}k}\right\}\text{erfc}\bigg(\frac{4\alpha_{0}\alpha_{1}k+4\alpha_{1}}{8\alpha_{1}^{3/2}k^{1/2}}\bigg).
\end{equation}
\end{thm}
\begin{proof}
\begin{equation*}
E(T^{k})=\int_{0}^{1}t^{k-1}(\alpha_{0}^{2}-4\alpha_{1}\log t)^{-1/2}\exp\left\{\frac{-\alpha_{0}+(\alpha_{0}^{2}-4\alpha_{1}\log(t))^{1/2}}{-2\alpha_{1}}\right\}dt.
\end{equation*}
If we define $u=\frac{-\alpha_{0}+(\alpha_{0}^{2}-4\alpha_{1}\log(t))^{1/2}}{-2\alpha_{1}}$, we have
\begin{equation*}
E(T^{k})=\exp\left\{ \alpha_{1}k\bigg(\ \frac{4\alpha_{0}\alpha_{1}k+4\alpha_{1}}{8\alpha_{1}^{2}k}\bigg)^{2}\right\}\int_{-\infty}^{0}\exp\left\{-\alpha_{1}k\bigg[u- \bigg(\ \frac{4\alpha_{0}\alpha_{1}k+4\alpha_{1}}{8\alpha_{1}^{2}k}\bigg)\bigg]^{2}\right\}du.
\end{equation*}
Simplifying further, we have the final result as
\begin{equation*}
E(T^{k})=\frac{1}{2}\sqrt{\frac{\pi}{\alpha_{1}k}}\exp\left\{\frac{(4\alpha_{0}\alpha_{1}k+4\alpha_{1})^{2}}{64\alpha_{1}^{3}k}\right\}\text{erfc}\bigg(\frac{4\alpha_{0}\alpha_{1}k+4\alpha_{1}}{8\alpha_{1}^{3/2}k^{1/2}}\bigg).
\end{equation*}
\end{proof}
With this result, we have $E(T)$ and $Var(T)$ as
\begin{align*}
  E(T)=&\frac{1}{2}\sqrt{\frac{\pi}{\alpha_{1}}}\exp\left\{\frac{(4\alpha_{0}\alpha_{1}+4\alpha_{1})^{2}}{64\alpha_{1}^{3}}\right\}\text{erfc}\bigg(\frac{4\alpha_{0}\alpha_{1}+4\alpha_{1}}{8\alpha_{1}^{3/2}}\bigg)\\
  Var(T)=&\frac{1}{2}\sqrt{\frac{\pi}{2\alpha_{1}}}\exp\left\{\frac{(8\alpha_{0}\alpha_{1}+4\alpha_{1})^{2}}{128\alpha_{1}^{3}}\right\}\text{erfc}\bigg(\frac{8\alpha_{0}\alpha_{1}+4\alpha_{1}}{(2^{7/3}\alpha_{1})^{3/2}}\bigg)\\
  &-\frac{1}{4}\frac{\pi}{\alpha_{1}}\exp\left\{\frac{(4\alpha_{0}\alpha_{1}+4\alpha_{1})^{2}}{32\alpha_{1}^{3}}\right\}\left\{\text{erfc}\bigg(\frac{4\alpha_{0}\alpha_{1}+4\alpha_{1}}{8\alpha_{1}^{3/2}}\bigg)\right\}^{2}.\\
\end{align*}

The median for the complementary extended power distribution is
\begin{equation*}
Q_{0.5}=\exp\left\{ \alpha_{0}\log(0.5)-\alpha_{1}(\log(0.5))^{2}\right\}.
\end{equation*}
and the mode is obtainable in a closed form by calculating the first derivative of $q(t)$ and equating to zero. Hence, we can calculate the mode by solving the cubic equation
\begin{equation}\label{eqn:eqn20}
A_{0}+A_{1}\log t+A_{2}(\log t)^{2}+A_{3}(\log t)^3=0
\end{equation}
where
\begin{eqnarray*}
  A_{0} &=& \alpha_{0}^{2}-4\alpha_{0}^{4}\alpha_{1}+4\alpha_{0}^{6}\alpha_{1}^{2}-1 \\
  A_{1} &=& -48\alpha_{0}^{4}\alpha_{1}^{3}-4\alpha_{1} \\
  A_{2} &=& 192\alpha_{0}^{2}\alpha_{1}^{4}-64\alpha_{1}^{3} \\
  A_{3} &=& -256\alpha_{1}^{5}.
\end{eqnarray*}
\subsection{Maximum Likelihood Estimation}
The log-likelihood function of the complementary extended power distribution is
\begin{equation}\label{eqn:eqn21}
\ell(\alpha_{0}, \alpha_{1})=\sum_{i=1}^{n}\log (\alpha_{0}^{2}-4\alpha_{1}\log t_{i})^{-1/2}-\sum_{i=1}^{n}\log t_{i}+\frac{\alpha_{0}n}{2\alpha_{1}}-\frac{1}{2\alpha_{1}}\sum_{i=1}^{n}(\alpha_{0}^{2}-4\alpha_{1}\log t_{i})^{1/2}
\end{equation}
and we have the following system of equations by differential the log-likelihood
\begin{eqnarray*}
  \frac{\partial \ell(\alpha_{0}, \alpha_{1})}{\partial \alpha_{0}} &=& -\alpha_{0}\sum_{i=1}^{n}(\alpha_{0}^{2}-4\alpha_{1}\log t_{i})^{-1}+\frac{n}{2\alpha_{1}}-\frac{\alpha_{0}}{2\alpha_{1}}\sum_{i=1}^{n}(\alpha_{0}^{2}-4\alpha_{1}\log t_{i})^{-1/2}=0 \\
\frac{\partial \ell(\alpha_{0}, \alpha_{1})}{\partial \alpha_{0}}   &=& 2\sum_{i=1}^{n}\log t_{i}(\alpha_{0}^{2}-4\alpha_{1}\log t_{i})^{-1}-\frac{\alpha_{0}n}{2\alpha_{1}^{2}}-\frac{\alpha_{1}^{-3}}{2}\sum_{i=1}^{n}(\alpha_{0}^{2}-4\alpha_{1}\log t_{i})^{-1/2}=0.
\end{eqnarray*}
Like in the extended power function, we can apply non-linear optimisation to estimate the unknown parameters.
\section{Conclusion}

A bounded probability distribution motivated by warping functions in functional data analysis is proposed. We have explored properties of the extended power distribution (EPD), such as moments and applications. We have given special cases of the distribution which are related to the Raleigh distribution, exponential distribution, beta distribution and linear hazard rate distribution (\cite{bain1974analysis}).
In this work, closed forms for the mode, median and other quantiles were given. An important property of this distribution over the beta distribution is that we have its cumulative distribution function and quantile function available in simple mathematical forms which makes it easy to simulate using the probability integral transform.
\begin{figure}[ht]
      \centering
    \includegraphics[scale=0.8]{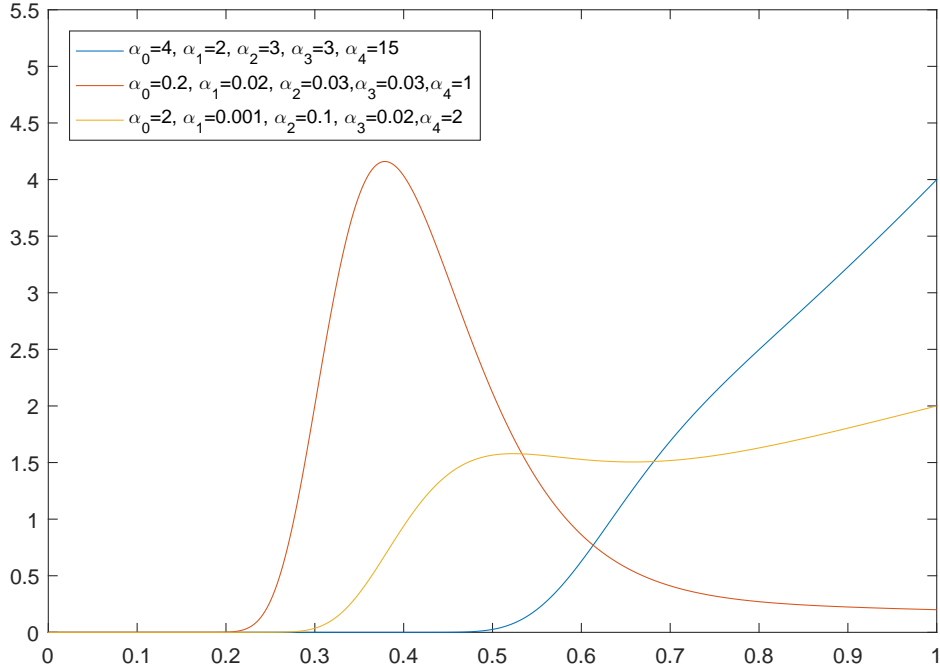}
    \caption{Plots of the 5-parameter extended power distribution for different parameter values showing some extra flexibilities}
    \label{fig:fig8}
  \end{figure}

We note that the Kumaraswamy distribution and extended power distribution have some interesting properties in common, such as a closed form for their cumulative distribution and quantile functions. However, unlike Kumaraswamy distribution, the extended power distribution is easily extendable from a two-parameter to a muilti-parameter distribution. This multi-parameter extension comes with added flexibility as we have seen in some applications in this work. In figure \ref{fig:fig8} for example, we show a five-parameter case of the extended power distribution with different parameter choices. Another advantage of the extended power distribution is that as $t$ nears $1$, the density approaches the parameter $\alpha_{0}$, while the density function of the Kumaraswamy distribution (and the beta distribution) approaches $0$ or $\infty$, as $t$ approaches $1$. This is useful in applications where there is a high proportion of observed values closer to the upper bound.

The generalisation of the extended power distribution makes computing moments more complicated and would involve numerical integration. However, simulations using the probability integral transform are less complicated and would involve finding the roots of a polynomial (same applies to the median). We can also estimate the parameters of the generalised extended power distribution numerically, using maximum likelihood.

Like in \cite{jones2002complementary}, we have proposed a complementary extended power distribution, which is the distribution derived from the quantile of the extended power distribution. We have shown that this distribution is linked to the beta distribution and the exponential distribution with inverted parameters. Simulation from this distribution, are easy using the probability integral transform because of the closed form of the cumulative density of the extended power distribution. However, the mode of the complementary extended power distribution is not available in closed form, but can the calculated as the solution to a cubic equation.
\begin{figure}[ht]
      \centering
    \includegraphics[scale=0.8]{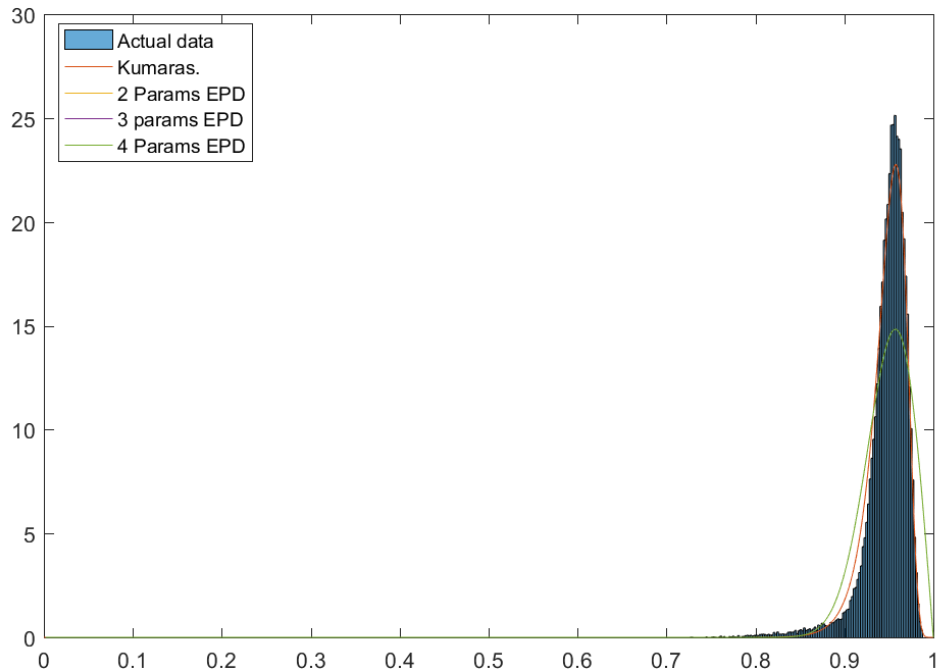}
    \caption{Fitting distributions to skewed data with high peak}
    \label{fig:fig12}
  \end{figure}

We note that for skewed data (either left or right skewed) with very high peaks, the Kumaraswamy distribution gives a better fit than both the beta distribution and extended power distribution. This is seen in figure \ref{fig:fig12}, for data containing information on employment rates in different counties in the US (\cite{BLSemploy}). The Kumaraswamy distribution has a higher peak than the two, three and four-parameter extended power distribution.

In the future, more applications peculiar to the extended power distribution needs to be explored. It will also be interesting to consider a new family of distributions by combining the extended power distribution and its complementary distribution. Just like the beta regression models are used for modelling bounded responses, we believe the extended power distribution can be alternative to the beta distribution. Programs for calculating different quantities related to extended power distribution are available in MATLAB.
\bibliographystyle{plainnat}
\bibliography{JSC}
\end{document}